\documentclass[11pt]{article}

\usepackage{latexsym}
\usepackage{theorem}
\usepackage{color,graphicx}
\usepackage{amssymb}
\usepackage{amsmath}
\usepackage{txfonts}
\usepackage{enumitem}
\usepackage{fullpage}

\newenvironment{proof}[1][Proof]
{\par\noindent{\bf #1:} }{\hspace*{\fill}\nolinebreak{$\Box$}\bigskip\par}
\newcommand{\qed}{\hspace*{\fill}\nolinebreak\ensuremath{\Box}}

\newtheorem{theorem}{Theorem}
\newtheorem{lemma}{Lemma}[section]

\newtheorem{observation}[lemma]{Observation}

\theorembodyfont{\upshape}

\newcommand{\cS}{\mathcal{S}}
\newcommand{\cP}{\mathcal{P}}

\newcommand{\cM}{\mathcal{M}}

\newcommand{\vectL}{\mathbf{L}}
\newcommand{\vectU}{\mathbf{U}}
\newcommand{\vectW}{\mathbf{W}}
\newcommand{\vectP}{\mathbf{P}}
\newcommand{\vectA}{\mathbf{A}}
\newcommand{\vectB}{\mathbf{B}}
\newcommand{\vectI}{\mathbf{I}}
\newcommand{\vecta}{\mathbf{\alpha}}
\newcommand{\vectb}{\mathbf{\beta}}
\newcommand{\T}{\textup{T}}

\newcommand{\Mshared}[1][]{\cM \if!#1!\else (#1) \fi}
\newcommand{\Mpriv}{\cP}

\newcommand{\complTime}[3]{C_{#1}^{#3}(#2)}  
\newcommand{\startTime}[3]{s_{#1}^{#3}(#2)}  
\newcommand{\tct}[1]{\varSigma(#1)}                  

\newcommand{\st}{\hspace{0.1cm}\bigl|\bigr.\hspace{0.1cm}}

\newcommand{\jobs}{\mathcal{J}}

\begin{document}

\title{\textbf{Shared multi-processor scheduling}}

\author{
  Dariusz Dereniowski\footnote{Corresponding author. Email: deren@eti.pg.gda.pl}\\
  \small{\emph{Faculty of Electronics,}}\\
  \small{\emph{Telecommunications and Informatics},}\\
  \small{\emph{Gda{\'n}sk University of Technology},}\\
  \small{\emph{Gda{\'n}sk, Poland}}
\and
  Wies{\l}aw Kubiak\\
  \small{\emph{Faculty of Business Administration},}\\
  \small{\emph{Memorial University},}\\
  \small{\emph{St. John's, Canada}}
}

\date{}

\maketitle

\begin{abstract}
We study shared multi-processor scheduling problem where each job can be executed on its private processor and simultaneously on one of many processors shared by all jobs in order to reduce the job's completion time due to processing time \emph{overlap}. The total weighted overlap of all jobs is to be maximized. The problem models subcontracting scheduling in supply chains and divisible load scheduling in computing. We show that synchronized schedules that complete each job at the same time on its private and shared processor, if any is actually used by the job, include optimal schedules. We prove that the problem is NP-hard in the strong sense for jobs with arbitrary weights, and we give an efficient, polynomial-time algorithm for the problem with equal weights.
\end{abstract}

\textbf{Keywords:} combinatorial optimization, divisible jobs, shared processors, subcontracting, supply chains

\section{Introduction}

The problem of scheduling divisible jobs on shared processors has attracted growing attention due to its  importance in scheduling job-shops, parallel and distributed computer systems, and supply chains.

Anderson~\cite{A81} considers a job-shop scheduling model where each job is a batch of potentially infinitely small items that can be processed independently of
other items of the batch. A processor in the shop is being shared between the jobs processed by the processor at the rates proportional to the
processor capacity fraction allocated to them by scheduler. The objective is to minimize total weighted backlog in a given time horizon.

Bharadwaj et. al.~\cite{HBGR03} survey divisible load scheduling where fractions of total divisible load are distributed to subsets of nodes 
of a shared network of processors for distributed processing. The processing by the nodes and possible communications
between the nodes overlap in time so that the completion time (makespan) for the whole load is shorter than the processing  of the whole load by a single node. The goal is to chose the size of the load fractions for each node so that the makespan for the whole load is minimized. \cite{HBGR03} points out that many real-life applications satisfy the divisibility property, among them "...processing of massive experimental
data, image processing applications like feature extraction
and edge detection, and signal processing applications like
extraction of signals buried in noise from multidimensional
data collected over large spans of time, computation of Hough
transforms, and matrix computations." Drozdowski~\cite{D09} surveys optimal solutions for a \emph{single} divisible load
obtained for various network topologies.

Recently, Vairaktarakis and Aydinliyim~\cite{VairaktarakisAydinliyim07} consider scheduling divisible jobs on subcontractor's processor  in supply chains to reduce the job's completion times. 
Hezarkhani and Kubiak~\cite{HK15} refer to the problem as the subcontractor scheduling problem.
Vairaktarakis~\cite{V13} points out that the lack of due attention to the subcontractors'
operations can cause significant complications in supply chains. A well-documented real-life example
of this issue has been reported in Boeing's Dreamliner
supply chain where the overloaded schedules of subcontractors,
each working with multiple suppliers, resulted in long
delays in the overall production due dates (see Vairaktarakis~\cite{V13} for more details and references). The subcontractor scheduling problem
is common in quick-response
industries characterized by volatile demand and inflexible
capacities where subcontracting is often used --- those include 
metal fabrication industry (Parmigiani~\cite{Par}), electronics assembly  (Webster et al~\cite{Web}), high-tech
manufacturing (Aydinliyim and Vairaktarakis~\cite{Ay11}), textile
production, and engineering services (Taymaz and Kili\c{c}aslan~\cite{Tay}) where subcontracting enables a manufacturer to
speed up the completion times of his jobs.

In the subcontractor scheduling problem
each agent has  its private processor  and a \emph{single} subcontractor's processor shared by all jobs available for the 
execution of its own job. The jobs can be divided between private and shared processor
 so that job completion times are reduced by possibly overlapping executions on private and shared processor.  
Vairaktarakis and Aydinliyim~\cite{VairaktarakisAydinliyim07} consider a non-preemptive case where  at most one time interval
on subcontractor's shared processor is allowed for any job. They prove that under this assumption there
exist optimal schedules that complete job execution on private and shared processor at the same time, we refer to such schedules
as \emph {synchronized} schedules, and show that sequencing jobs in ascending order of their processing times on the shared
processor gives an optimal solution. Furthermore this solution guarantees  non-empty interval  on the shared processor  for each
job. Hezarkhani and Kubiak \cite{HK15} observe that by allowing an agent to use a set of several
mutually disjoint intervals on the subcontractor processor one does not improve schedules by increasing total overlap. Therefore,  \cite{HK15} actually observes that algorithm of
\cite{VairaktarakisAydinliyim07} solves the single processor preemptive problem to optimality as well. In this paper we generalize this preemptive model of \cite{HK15} by allowing many shared processors
and by allowing  that the reduction in job completion time be rewarded at different rates for different jobs, i.e., we allow different weights for jobs.

It is worth pointing out that 
Vairaktarakis and Aydinliyim \cite{VairaktarakisAydinliyim07} change focus from optimization typically sought after in the centralized setting to
coordinating mechanisms to ensure efficiency in the decentralized systems. 
 Vairaktarakis  \cite{V13} analyzes the outcomes of a
decentralized subcontracting system under different protocols
announced by the subcontractor. Both papers
assume complete information yet neither provides coordinating
pricing schemes for the problem. To remedy this \cite{HK15}  designs parametric pricing schemes that strongly
coordinate this decentralized system with complete information, that is, they ensure that the agents' choices
of subcontracting intervals always result in efficient (optimal) schedules. It also proves that the pivotal mechanism
is coordinating, i.e., agents are better off by reporting their true processing times, and by participating in the subcontracting.

The remainder of the paper is organized as follows. Section 2 introduces notation and formulates the shared multi-processor scheduling problem. Section 3 defines some desirable characteristics of schedules and proves that there always are optimal schedules with these characteristics. Section 4 proves that there always is an optimal schedule that is synchronized.
Section 5 considers special instances for which optimal schedules on shared processors are $V$-\emph{shaped} and \emph{reversible}. Section 6 proves that the problem is NP-hard in the strong sense
even when limited to the set of instances defined in Section 5.  Section 7 gives an efficient, polynomial time algorithm for the problem with equal weights. Finally, Section 8 concludes the paper and lists  open problems.

\section{Problem formulation} \label{sec:problem}
We are given a set $\jobs$ of $n$ preemptive jobs.
Each job $j\in\jobs$ has its processing time $p_j$ and weight $w_j$.
With each job $j\in\jobs$ we associate its \emph{private} processor denoted by $\Mpriv_j$.
Moreover, $m\geq 1$ \emph{shared} processors $\Mshared_1,\ldots,\Mshared_m$ are available for all jobs.

A \emph{feasible} schedule $\cS$ selects for each job $j\in\jobs$:
\begin{enumerate} [label={\normalfont{(\roman*)}}]
 \item\label{it:f1} a shared processor $\Mshared[\cS,j]\in \{\mathcal{M}_1, \ldots, \mathcal{M}_m\}$,
 \item\label{it:f2} a (possibly empty) \emph{set} of open, mutually disjoint time intervals in which $j$ executes on $\Mshared[\cS,j]$, and
 \item\label{it:f3} a \emph{single} time interval $(0, \complTime{\cS}{j}{\Mpriv})$ where $j$ executes on its private processor $\Mpriv_j$.
\end{enumerate}
The total length of all these intervals (the ones in~\ref{it:f2} and the one in~\ref{it:f3}) equals $p_j$. The simultaneous execution of $j$ on private $\Mpriv_j$ and shared $\Mshared[\cS,j]$ is allowed and desirable, as follows from the optimization criterion given below. However, for any two jobs $j$ and $j'$ if they use the same shared processor, i.e.,  $\Mshared[\cS,j]=\Mshared[\cS,j']=\Mshared$, then any interval in which $j$ executes on $\Mshared$ is disjoint from any interval in which $j'$ executes on $\Mshared$.
In other words, each processor can execute at most one job at a time.

Given a feasible schedule $\cS$, for each job $j\in\jobs$ we call any time interval of maximum length in which $j$ executes on both  private $\Mpriv_j$ and shared $\Mshared[\cS,j]$ simultaneously an \emph{overlap}.
The total overlap $t_j$ of  job $j$ equals the sum of lengths of all overlaps for $j$.
The \emph{total weighted overlap} of $\cS$ equals
\[\tct{\cS}=\sum_{j\in\jobs}t_jw_j.\] 
A feasible schedule that maximizes the total weighted overlap is called \emph{optimal}. For convenience we use the abbreviation WSMP to denote
the \emph{weighted shared multi-processor} scheduling problem: the instance of the problem consists of a set of jobs $\jobs$ and the number of shared processors $m$; the goal is to find an optimal schedule that maximizes total weighted overlap. 

This objective function is closely related to the total completion time objective traditionally used in scheduling. The total completion time \emph{can} potentially be reduced by an increase of the total overlap resulting
from the simultaneous execution of jobs on private and shared processors. However, to take full advantage of this potential the schedules need to start jobs \emph{at} time $0$, otherwise the 
overlap would not necessarily be advantageous in reducing the total completion time. At the same time we need to emphasize that the two objectives exist for different practical reasons. The minimization of total completion time 
minimizes mean flow time and thus by Little's Law minimizes average inventory in the system. The maximization of the total overlap on the other hand maximizes the total net payoff resulting from completing job $j$ earlier at $p_j-t_j$ thanks to the use of shared processors (subcontractors) rather than at $p_j$ if those where not used. The $w_jt_j$ is a net payoff obtained from the completion of job (order) $t_j$ time units earlier due to the overlap $t_j$. This different focus sets the total weighted overlap objective apart from the total completion time objective as an objective important in practice in scheduling shared processors.

For illustration let us consider an example in Figure~\ref{fig:example} with two shared processors and $6$ jobs.
Note that in this example, each job completes at the same time on its private processor and on a shared one (Sections~\ref{sec:observations} and~\ref{sec:optimal} will conclude that for each problem instance there exists an optimal solution with this property).
\begin{figure}[htb]
\begin{center}
\includegraphics[scale=0.9]{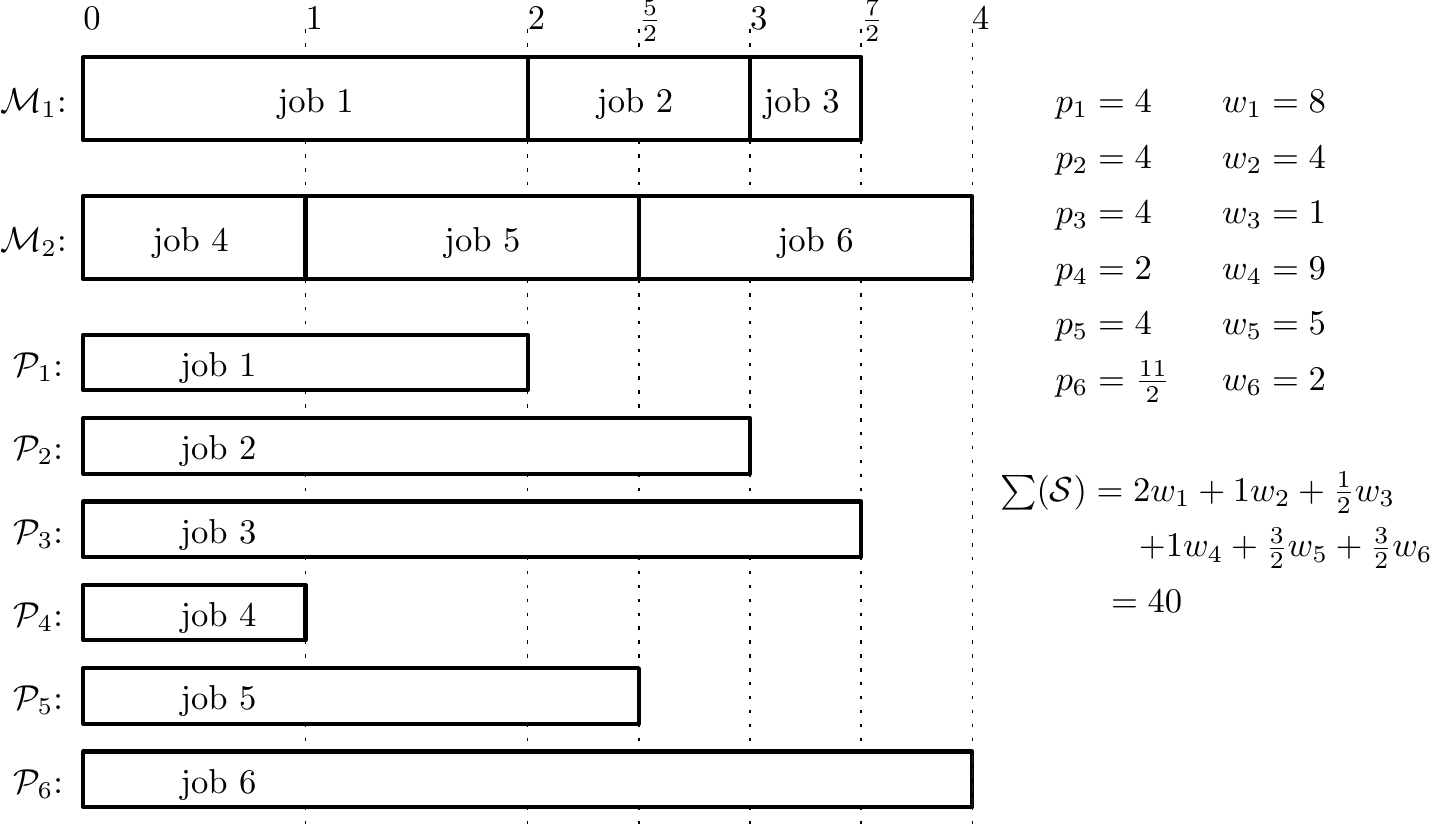}
\caption{A schedule for six-job instance and $m=2$ shared processors.}
\label{fig:example}
\end{center}
\end{figure}

\section{Simple observations  about optimal schedules} \label{sec:observations}

We now make four simple observations that allow us to reduce a class of schedules to consider yet ensure at the same time that the reduced class always includes optimal schedules.
Let $\cS$ be a feasible schedule. Let $\startTime{\cS}{j}{\Mshared}$ and $\complTime{\cS}{j}{\Mshared}$ be the start time and the completion times of a job $j$ on the shared processor $\Mshared[\cS,j]$ respectively, both being $0$ if all of $j$ executes on its private processor only.
A schedule $\cS$ is called \emph{normal} if $\complTime{\cS}{j}{\Mshared}\leq\complTime{\cS}{j}{\Mpriv}$ for each $j\in\jobs$.
We observe the following.
\begin{observation} \label{obs:normal}
There exists an optimal schedule that is normal.
\end{observation}
\begin{proof}
Let $\cS$ be an optimal schedule.
Suppose that some job $j$ completes on a shared processor later than on its private processor in $\cS$, i.e., $\complTime{\cS}{j}{\Mshared}>\complTime{\cS}{j}{\Mpriv}$ and thus $\cS$ is not normal.
Hence, there exist intervals $I_1,\ldots,I_k$ such that for each $i\in\{1,\ldots,k\}$, the processor $\Mshared[\cS,j]$ executes $j$ in $I_i$, $I_i\subseteq(\complTime{\cS}{j}{\Mpriv},+\infty)$ and no part of $j$ executes in $(\complTime{\cS}{j}{\Mpriv},+\infty)\setminus\bigcup_{i=1}^kI_i$ on either $\Mshared[\cS,j]$ or $\Mpriv_j$.
Then, modify $\cS$ by removing the job $j$ from all intervals $I_1,\ldots,I_k$ on the shared processor $\Mshared[\cS,j]$ (so that $\Mshared[\cS,j]$ is idle in $\bigcup_{i=1}^k I_i$) and let $j$ execute in the interval $(0,\complTime{\cS}{j}{\Mpriv}+\sum_{i=1}^k|I_i|)$ on its private processor $\Mpriv_j$.
Note that the total weighted overlap of $\cS$ has not changed by this transformation.
After repeating this transformation for each job $j$ if need be we obtain an optimal schedule that is normal.
\end{proof}
\begin{observation} \label{obs:idle-time}
Let $\cS$ be an optimal normal schedule and let $X_i=\{j\in\jobs\st \Mshared[\cS,j]=\Mshared_i \text{ and }  \complTime{\cS}{j}{\Mshared}>0\}$ for each $i\in\{1,\ldots,m\}$.
There is no idle time in time interval $(0,\max\{\complTime{\cS}{j}{\Mshared}\st j\in X_i\})$ on each shared processor $\Mshared_i$.
\end{observation}
\begin{proof}
Note that by Observation~\ref{obs:normal}, there exists a normal optimal schedule $\cS$.
Suppose for a contradiction that some shared processor $\Mshared_i$ is idle in a time interval $(l,r)\neq\emptyset$ and $r<\complTime{\cS}{j}{\Mshared}\leq \complTime{\cS}{j}{\Mpriv}$ for some job $j\in X_i$.
Take maximum $\varepsilon\in(0,(r-l)/2]$ such that $j$ executes continuously in $I_1=(\complTime{\cS}{j}{\Mpriv}-\varepsilon,\complTime{\cS}{j}{\Mpriv})$ on $\Mpriv_j$ and in $I_2=(\complTime{\cS}{j}{\Mshared}-\varepsilon,\complTime{\cS}{j}{\Mshared})$ on $\Mshared_i$.
Then, obtain a schedule $\cS'$ by taking a piece of $j$ that executes in $I_1$ on $\Mpriv_j$ and a piece of $j$ that executes in $I_2$ on $\Mshared_i$ and execute both pieces in $(l,r)$ on $\Mshared_i$.
Clearly, the new schedule $\cS'$ is feasible and, since $\cS$ is normal,
$\tct{\cS'}=\tct{\cS}+w_j\varepsilon$, which contradicts the optimality of $\cS$.
\end{proof}

We say that a schedule $\cS$ is \emph{non-preemptive} if each job $j$ executes 
in time interval $(\startTime{\cS}{j}{\Mshared},\complTime{\cS}{j}{\Mshared})$ on
$\Mshared[\cS,j]$ in $\cS$.
In other words, in a non-preemptive schedule there is at most one interval in~\ref{it:f2} in the definition of a feasible schedule.
\begin{observation} \label{obs:non-preemptive}
There exists an optimal schedule that is normal and non-preemptive.
\end{observation}
\begin{proof}
By Observation~\ref{obs:normal}, there exists a normal optimal schedule $\cS$.
Suppose that $\cS$ is preemptive.
We transform $\cS$ into a non-preemptive one whose total weighted overlap is not less than that of $\cS$.
The transformation is performed iteratively.
At the beginning of each iteration a job $j$ is selected such that $j$ executes on a shared processor $\Mshared[\cS,j]$ in at least two disjoint time intervals 
 $(l,r)$ and $(l',r')$, where $r<l'$. Without loss of generality we assume that the intervals are of maximal lengths.
Modify $\cS$ by shifting each job start, completion and preemption that occurs in time interval $(r,l')$ on $\Mshared[\cS,j]$ by $r-l$ units to the left, i.e. towards the start of the schedule at 0.
Then, the part of $j$ executed in $(l,r)$ in $\cS$ is executed in $(l'-r+l,l')$ after the transformation.
The transformation does not increase the completion time of any job $j'$ on $\Mshared[\cS,j]$  and keeps it the same on $\Mpriv_{j'}$ for each job $j'$.
Thus, in particular, $\cS$ remains normal.
However, the number of preemptions of the job $j$ decreases by $1$ and there is no job whose number of preemptions increases.
Also, the total weighted overlap of $\cS$ does not change.
Hence, after finite number of such iterations we arrive at a required normal non-preemptive optimal schedule.
\end{proof}

We say that a schedule $\cS$ is \emph{ordered} if it is normal, non-preemptive and for any two jobs $j$ and $j'$ assigned to the same shared processor it holds $\complTime{\cS}{j}{\Mshared}\leq\complTime{\cS}{j'}{\Mshared}$ if and only if $\complTime{\cS}{j}{\Mpriv}\leq\complTime{\cS}{j'}{\Mpriv}$.
Informally speaking, the order of job completions on the shared processors is the same as the order of their completions on the private processors.

\begin{observation} \label{lem:ordered}
There exists an optimal  schedule that is ordered.
\end{observation}
\begin{proof}
Let $\cS$ be an optimal normal and non-preemptive schedule; such a schedule exists due to Observation~\ref{obs:non-preemptive}.
Let $X_i=\{j\in\jobs\st \Mshared[\cS,j]=\Mshared_i \text{ and }  \complTime{\cS}{j}{\Mshared}>0\}$ for each $i\in\{1,\ldots,m\}$.
Recall that each job $j\in X_i$ executes in a single interval $(\startTime{\cS}{j}{\Mshared},\complTime{\cS}{j}{\Mshared})$ on $\mathcal{M}_i$ in a non-preemptive $\cS$.
By Observation~\ref{obs:idle-time}, there is no idle time  in time interval $(0,\complTime{\cS}{j}{\Mshared})$ for each job $j\in X_i$ on $\Mshared_i$.
Thus, we may represent $\cS$ on a processor $\Mshared_i$ as a sequence of pairs $\Mshared_i=((j_1,l_1),\ldots,(j_k,l_k))$, where $X_i=\{j_1,\ldots,j_k\}$ and the job $j_t$ executes in time interval
\[\left( \sum_{j'=0}^{t-1}l_{j'}, \sum_{j'=0}^{t}l_{j'} \right) \]
on $\Mshared_i$, where $l_0=0$, and 
in time interval $(0,p_{j_t}-l_t)$ on $\Mpriv_{j_t}$ for each $t\in\{1,\ldots,k\}$.

If $\cS$ is ordered, then the proof is completed.
Hence, suppose that $\cS$ is not ordered.
There exists a shared processor $\Mshared_i$ and an index $t\in\{1,\ldots,k-1\}$ such that
\begin{equation} \label{eq:ordered:normal}
\complTime{\cS}{j_t}{\Mshared}<\complTime{\cS}{j_{t+1}}{\Mshared} \quad\textup{and}\quad \complTime{\cS}{j_t}{\Mpriv}>\complTime{\cS}{j_{t+1}}{\Mpriv}.
\end{equation}
Consider a new non-preemptive schedule $\cS'$ in which:
\[\Mshared_i=((j_1,l_1),\ldots,(j_{t-1},l_{t-1}),(j_{t+1},l_{t+1}),(j_t,l_t),(j_{t+2},l_{t+2}),\ldots,(j_k,l_k)),\]
i.e., the order of jobs $j_t$ and $j_{t+1}$ has been reversed on $\Mshared_i$ while the schedules on all other processors remain unchanged.
Note that this exchange does not affect start times and completion times of any job on $\Mshared_i$ except for $j_t$ and $j_{t+1}$.
Since $\cS$ is normal, we obtain
\[\complTime{\cS'}{j_{t+1}}{\Mshared}\leq\complTime{\cS}{j_{t+1}}{\Mshared}\leq\complTime{\cS}{j_{t+1}}{\Mpriv}=\complTime{\cS'}{j_{t+1}}{\Mpriv}\]
and, also by \eqref{eq:ordered:normal},
\[\complTime{\cS'}{j_{t}}{\Mshared}=\complTime{\cS}{j_{t+1}}{\Mshared}\leq\complTime{\cS}{j_{t+1}}{\Mpriv}<\complTime{\cS}{j_t}{\Mpriv},\]
which proves that $\cS'$ is normal.
Clearly, $\tct{\cS'}=\tct{\cS}$.
Set $\cS:=\cS'$ and repeat the exchange if need be.
After a finite number of such exchanges we arrive at a schedule that is ordered.
\end{proof}

\section{Optimal schedules are synchronized} \label{sec:optimal}
We say that a schedule is \emph{synchronized} if it is normal, non-preemptive and for each job $j$ whose part executes on some shared processor it holds $\complTime{\cS}{j}{\Mshared}=\complTime{\cS}{j}{\Mpriv}$.
Note that a synchronized schedule is also ordered but the reverse implication does not hold in general.

In order to prove that there are optimal schedules that are synchronized we introduce \emph{pulling} and \emph{pushing} schedule transformations.
Let $\cS$ be an optimal ordered (possibly synchronized) schedule.
Consider a shared processor $\Mshared_r$.
Let $X_r=\{j_1,\ldots,j_k\}$, $k>1$, be jobs executed on  $\Mshared_r$ in $\cS$ 
and ordered according to increasing order of their completion times on $\Mshared_r$.
Let $i\in\{2,\ldots,k\}$ be an index such that $\complTime{\cS}{j_{\ell}}{\Mshared}=\complTime{\cS}{j_{\ell}}{\Mpriv}$ for each $\ell\in\{i,\ldots,k\}$.
(Recall that $\complTime{\cS}{j_{i}}{\Mshared}\leq\complTime{\cS}{j_{i}}{\Mpriv}$ since $\cS$ is normal.)
Observe that $j_k$ completes at the same time on its private processor (since $\cS$ is optimal) and on $\Mshared_r$ and hence the index $i$ is well defined.
Finally, let
\[0<\varepsilon\leq\startTime{\cS}{j_i}{\Mshared}-\startTime{\cS}{j_{i-1}}{\Mshared} = \complTime{\cS}{j_{i-1}}{\Mshared}-\startTime{\cS}{j_{i-1}}{\Mshared}.\]
We define an operation of \emph{pulling of} 
$j_{i}$ by $\varepsilon$ in $\cS$ as a transformation of $\cS$ that results in a schedule $\cS'$ defined as follows.
First, $\cS$ and $\cS'$ are identical on $\Mshared_r$ in time interval $(0,\startTime{\cS}{j_i}{\Mshared}-\varepsilon)$.
Then, for the job $j_{i-1}$ we set:
\[\complTime{\cS'}{j_{i-1}}{\Mshared}=\complTime{\cS}{j_{i-1}}{\Mshared}-\varepsilon, \quad\textup{and}\quad\complTime{\cS'}{j_{i-1}}{\Mpriv}=\complTime{\cS}{j_{i-1}}{\Mpriv}+\varepsilon.\]
Next, for each $\ell\in\{i,\ldots,k\}$ (by proceeding with subsequent increasing values of $\ell$) we define how $j_{\ell}$ is executed in $\cS'$:
\[\startTime{\cS'}{j_{\ell}}{\Mshared}=\complTime{\cS'}{j_{\ell-1}}{\Mshared}, \quad \complTime{\cS'}{j_{\ell}}{\Mshared}=\complTime{\cS}{j_{\ell}}{\Mshared}-\varepsilon/2^{\ell-i+1}, \quad \complTime{\cS'}{j_{\ell}}{\Mpriv}=\complTime{\cS}{j_{\ell}}{\Mpriv}-\varepsilon/2^{\ell-i+1}.\]
Finally, $\cS'$ and $\cS$ are identical on all other processors, i.e., on all processors different from $\Mshared_r$ and $\Mpriv_{j_{i-1}},\Mpriv_{j_i},\ldots,\Mpriv_{j_l}$. The operation of pulling $j_3$ by $\varepsilon$ is illustrated in Figure \ref{fig:pulling}.
Note that if we take $\varepsilon=\startTime{\cS}{j_i}{\Mshared}-\startTime{\cS}{j_{i-1}}{\Mshared}$, i.e., $\varepsilon$ equals the length of the entire execution interval of $j_{i-1}$ on $\Mshared_r$, then pulling of $j_i$ by $\varepsilon$ produces $\cS'$ in which $j_{i-1}$ executes only on its private processor. From this definition we have.

\begin{figure}[htb]
\begin{center}
\includegraphics[scale=1]{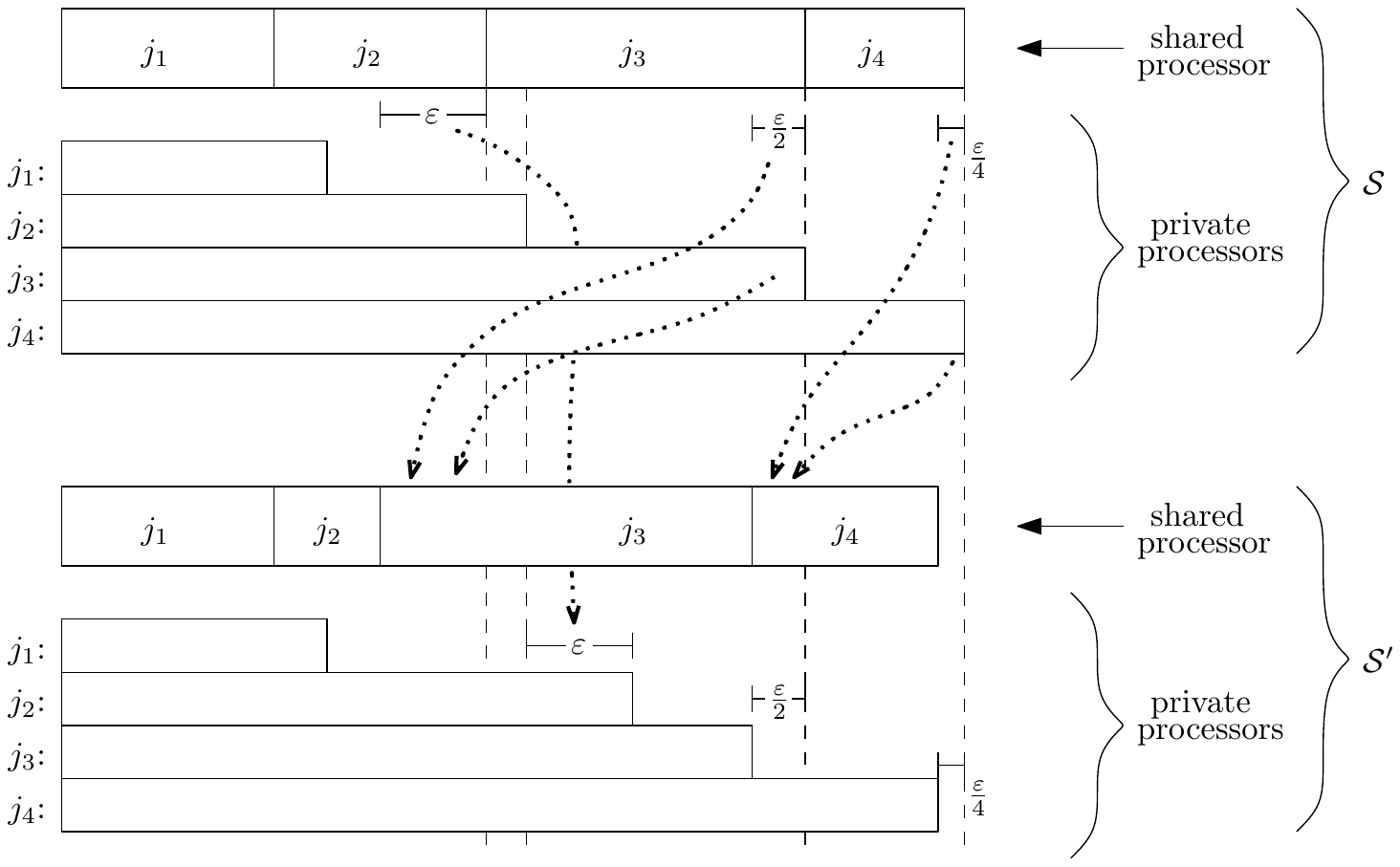}
\caption{An example of pulling of $j_3$ by $\varepsilon$ in some schedule $\cS$.}
\label{fig:pulling}
\end{center}
\end{figure}

\begin{lemma} \label{lem:pulling}
The pulling of $j_i$ by $\varepsilon$ in $\cS$ produces a feasible schedule $\cS'$ and
\[\tct{\cS'}=\tct{\cS}-\varepsilon w_{j_{i-1}}+\varepsilon\sum_{\ell=i}^k \frac{w_{j_{\ell}}}{2^{\ell-i+1}}.\]
\qed
\end{lemma}

We also define a transformation that is a `reverse' of pulling.
Recall that $j_1,\ldots,j_k$ are the jobs executing on the processor $\Mshared_r$.
For this transformation we assume that the schedule $\cS$ is ordered but not synchronized and some job completes on $\Mshared_r$ earlier than on its private processor.
Select $i\in\{2,\ldots,k\}$ to be the index such that $\complTime{\cS}{j_{\ell}}{\Mshared}=\complTime{\cS}{j_{\ell}}{\Mpriv}$ for each $\ell\in\{i,\ldots,k\}$
and $\complTime{\cS}{j_{i-1}}{\Mpriv}>\complTime{\cS}{j_{i-1}}{\Mshared}$.
The index $i$ is well defined because $\complTime{\cS}{j_k}{\Mshared}=\complTime{\cS}{j_k}{\Mpriv}$ in optimal schedule $\cS$.
Let
\begin{equation} \label{eq:epsp1}
0 < \varepsilon \leq \frac{1}{2}\left(\complTime{\cS}{j_{i-1}}{\Mpriv}-\complTime{\cS}{j_{i-1}}{\Mshared}\right)
\end{equation}
and for each $\ell\in\{i,\ldots,k\}$,
\begin{equation} \label{eq:epsp2}
\frac{\varepsilon}{2^{\ell-i+1}} \leq \complTime{\cS}{j_{\ell}}{\Mshared}-\startTime{\cS}{j_{\ell}}{\Mshared}.
\end{equation}
The transformation of \emph{pushing of} $j_i$ by $\varepsilon$ in $\cS$ produces schedule $\cS'$ defined as follows.
Both $\cS'$ and $\cS$ are identical on $\Mshared_r$ in time interval $(0,\startTime{\cS}{j_{i-1}}{\Mshared})$,
\[\complTime{\cS'}{j_{i-1}}{\Mshared}=\complTime{\cS}{j_{i-1}}{\Mshared}+\varepsilon\textup{ and }\complTime{\cS'}{j_{i-1}}{\Mpriv}=\complTime{\cS}{j_{i-1}}{\Mpriv}-\varepsilon.\]
Then, for each $\ell\in\{i,\ldots,k\}$ (with increasing values of $\ell$) we have
\[\startTime{\cS'}{j_{\ell}}{\Mshared}=\complTime{\cS'}{j_{\ell-1}}{\Mshared}, \quad \complTime{\cS'}{j_{\ell}}{\Mshared}=\complTime{\cS}{j_{\ell}}{\Mshared}+\frac{\varepsilon}{2^{\ell-i+1}}\]
and
\[\complTime{\cS'}{j_{\ell}}{\Mpriv}=\complTime{\cS}{j_{\ell}}{\Mpriv}+\frac{\varepsilon}{2^{\ell-i+1}}.\]
On each shared processor different than $\Mshared_r$ and on private processors different than $\Mpriv_{j_{i-1}},\ldots,\Mpriv_{j_k}$ the schedules $\cS$ and $\cS'$ are the same.

Note that if $\varepsilon/2^{\ell-i+1}=\complTime{\cS}{j_{\ell}}{\Mshared}-\startTime{\cS}{j_{\ell}}{\Mshared}$ for some $\ell\in\{i,\ldots,k\}$, then the pushing operation eliminates $j_{\ell}$ from the shared processor, i.e., $j_{\ell}$ executes only on its private processor in $\cS'$. From this definition we have.
\begin{lemma} \label{lem:pushing}
The pushing of $j_i$ by $\varepsilon$ in $\cS$ produces a feasible schedule $\cS'$ and
\[\tct{\cS'}=\tct{\cS}+\varepsilon w_{j_{i-1}}-\varepsilon\sum_{\ell=i}^k \frac{w_{j_{\ell}}}{2^{\ell-i+1}}.\]
\qed
\end{lemma}
We also note that if one first makes pulling of some job $j'$ by $\varepsilon$ in a schedule $\cS$ which results in a schedule $\cS'$ in which the same job precedes $j'$ on the shared processor both in $\cS$ and $\cS'$, then pushing of $j'$ by $\varepsilon$ in $\cS'$ results in returning back to $\cS$.

We are now ready to prove our main result of this section, which will allow us to work only with synchronized schedules in the sections that follow.
\begin{lemma} \label{lem:synchronized}
There exists an optimal synchronized schedule.
\end{lemma}
\begin{proof}
Let $\cS$ be an optimal schedule. By Observation~\ref{lem:ordered}, we may assume without loss of generality that $\cS$ is ordered.
Suppose that $\cS$ is not synchronized. We will convert $\cS$ into a synchronized schedule by iteratively performing transformations described below.

Let $\Mshared_r$ be a shared processor such that there exists a job assigned to $\Mshared_r$ that completes earlier on $\Mshared_r$ than on its private processor.
Let $X_r=\{j_1,\ldots,j_k\}$, $k>1$, be jobs executed on  $\Mshared_r$ in $\cS$ 
and ordered according to increasing order of their completion times on $\Mshared_r$.
Let $i\in\{2,\ldots,k\}$ be the minimum index such that $\complTime{\cS}{j_{\ell}}{\Mshared}=\complTime{\cS}{j_{\ell}}{\Mpriv}$ for each $\ell\in\{i,\ldots,k\}$.
Since $\cS$ is not synchronized and $\complTime{\cS}{j_k}{\Mshared}=\complTime{\cS}{j_k}{\Mpriv}$ in an optimal schedule, the index $i$ is well defined and $\complTime{\cS}{j_{i-1}}{\Mshared}<\complTime{\cS}{j_{i-1}}{\Mpriv}$ by the minimality of $i$.
We first argue that
\begin{equation} \label{eq:synchronized:ineq}
 w_{j_{i-1}}\geq\sum_{\ell=i}^k \frac{w_{j_{\ell}}}{2^{\ell-i+1}}.
\end{equation}
Consider pulling of $j_i$ by $\varepsilon=\complTime{\cS}{j_{i-1}}{\Mshared}-\startTime{\cS}{j_{i-1}}{\Mshared}$ in $\cS$ that produces a schedule $\cS'$.
Then, by Lemma~\ref{lem:pulling} and the optimality of $\cS$,
\[\tct{\cS'}=\tct{\cS}-\varepsilon w_{j_{i-1}}+\varepsilon\sum_{\ell=i}^k \frac{w_{j_{\ell}}}{2^{\ell-i+1}}
            =\tct{\cS}-\varepsilon\left( w_{j_{i-1}}-\sum_{\ell=i}^k \frac{w_{j_{\ell}}}{2^{\ell-i+1}}\right)\leq\tct{\cS},\]
which proves~\eqref{eq:synchronized:ineq}.

Following~\eqref{eq:epsp1} and~\eqref{eq:epsp2}, define
\[\varepsilon'=\min\left\{\left(\complTime{\cS}{j_{i-1}}{\Mpriv}-\complTime{\cS}{j_{i-1}}{\Mshared}\right)/2,
               \min\left\{2^{\ell-i+1}\left(\complTime{\cS}{j_{\ell}}{\Mshared}-\startTime{\cS}{j_{\ell}}{\Mshared}\right)\st \ell=i,\ldots,k\right\}\right\}.\]
Since $\cS$ is normal, by the choice of $i$ we have $\varepsilon'>0$.
Obtain a schedule $\cS'$ by performing pushing of $j_i$ by $\varepsilon'$ in $\cS$.
Note that if $\varepsilon'=(\complTime{\cS}{j_{i-1}}{\Mpriv}-\complTime{\cS}{j_{i-1}}{\Mshared})/2$, then $j_{i-1}$ completes at the same time on the shared and private processors in $\cS'$.
If, on the other hand, $\varepsilon'=2^{\ell-i+1}(\complTime{\cS}{j_{\ell}}{\Mshared}-\startTime{\cS}{j_{\ell}}{\Mshared})$ for some $\ell\in\{i,\ldots,k\}$, then $j_{\ell}$ is eliminated from the shared processor, i.e., $j_{\ell}$ executes only on its private processor in $\cS'$.
By Lemma~\ref{lem:pushing} and \eqref{eq:synchronized:ineq},
\[\tct{\cS'}=\tct{\cS}+\varepsilon' w_{j_{i-1}}- \varepsilon'\sum_{\ell=i}^k \frac{w_{j_{\ell}}}{2^{\ell-i+1}}
            =\tct{\cS}+\varepsilon'\left( w_{j_{i-1}}-\sum_{\ell=i}^k \frac{w_{j_{\ell}}}{2^{\ell-i+1}}\right)\geq\tct{\cS}.\]
Moreover, $\cS'$ satisfies the following two conditions:
\begin{enumerate} [label={\normalfont{(\alph*)}}]
 \item\label{eq:item:syn1} for each $\ell\in\{i,\ldots,k\}$, if $j_{\ell}$ is assigned to $\Mshared_r$ in $\cS'$, then $\complTime{\cS'}{j_{\ell}}{\Mshared}=\complTime{\cS'}{j_{\ell}}{\Mpriv}$,
 \item\label{eq:item:syn2} if all jobs $j_{\ell},\ldots,j_k$ are assigned to $\Mshared_r$ in $\cS'$, then $\complTime{\cS'}{j_{i-1}}{\Mshared}=\complTime{\cS'}{j_{i-1}}{\Mpriv}$.
\end{enumerate}
Set $\cS:=\cS'$ and repeat the transformation.

Condition~\ref{eq:item:syn1} ensures that if a job completes at the same time on private and shared processors in $\cS$, then this property is either preserved by the transformation or the job is executed only on its private processor in the new schedule $\cS'$.
Note that in the latter case, such a job will remain on its private processor during future transformations, never `returning' back to any shared processor; this follows directly from the pushing transformation.
Thus, in each transformation either  the number of jobs executing on shared processors decreases or, due to \ref{eq:item:syn2}, if this number does not decrease, then the number of jobs that complete at the same time on private and shared processors increases by one.
Hence it follows that after at most $2|\jobs|$ transformations we obtain an optimal schedule that is synchronized.
\end{proof}

In a synchronized schedule the order $j_1,\ldots,j_k$ of job executions on a processor $\Mshared_r$ uniquely determines the schedule on $\Mshared_r$ $r\in\{1,\ldots,m\}$.

\section{V-shapeness and duality of some instances} \label{sec:Vshape}
Our main goal in the section is to introduce special classes of instances that will provide a key to the complexity analysis of the problem in the next section. The following observation was made by \cite{VairaktarakisAydinliyim07} for a single shared processor non-preemptive problem and extended to preemptive one in \cite{HK15}. It will be used often in the remainder of the paper.
\begin{observation}
\label{obs:T_i}
If jobs $j_1,\ldots,j_k$ with processing times $p_1,\ldots,p_k$, respectively, are executed on a shared processor in some synchronized schedule in that order,
then the job $j_i$ executes in time interval $(T_i,T_{i+1})$ of length $\bar{t_i}$, where $T_1=0$ and
\[T_i=\sum_{\ell=1}^{i-1} \frac{p_{\ell}}{2^{i-\ell}}\textup{ for each }i\in\{2,\ldots,k+1\}, \quad \bar{t_i}=\frac{p_i}{2}-\sum_{\ell=1}^{i-1}\frac{p_{\ell}}{2^{i-\ell+1}}=\frac{p_i}{2}-\frac{T_i}{2}\textup{ for each }i\in\{1,\ldots,k\}.\]
\qed
\end{observation}

\subsection{Consecutive job exchange for processing-time-inclusive instances} \label{lem:V:general}

 We begin with a lemma which allows us to calculate the difference in total weighted overlaps of two schedules, one of which is obtained from the other by exchanging two consecutive jobs on a shared processor.
This exchange is complicated by the fact that for the job that gets later in the permutation after the exchange  it may no longer be possible to execute on the shared processor  since the job may prove too short for that. Generally, the test whether this actually happens depends not only on the processing times of the jobs that precede the later job but also on their \emph{order}. Instead, we
would like to be able to select an arbitrary subset $A$ of $\jobs$, take any permutation of the jobs in $A$, and always guarantee that there exists  a synchronized schedule that has exactly the jobs in $A$  that appear in the order determined by the permutation on the shared processor. Clearly, this freedom cannot be ensured for arbitrary instances. Therefore we introduce
an easy to test sufficient condition that would always guarantee the validity of the job exchange.

Consider a set of jobs $A=\{j_1,\ldots,j_k\}$, where we assume $p_{1}\leq p_{2}\leq\cdots\leq p_{k}$; here we take $p_i$ to be the processing time of the job $j_i$, $i\in\{1,\ldots,k\}$.
We say that the set of jobs $A$ is \emph{processing-time-inclusive} if
\[x=T_{|A|}=\sum_{\ell=1}^{|A|-1} \frac{p_{\ell+1}}{2^{|A|-\ell}} < p_1.\]
Note that $x$ is the makespan of a schedule on a shared processor for jobs in $A\setminus\{j_1\}$ when the jobs are scheduled in ascending order of their processing times, i.e., the order $j_2,\ldots,j_k$.
By~\cite{VairaktarakisAydinliyim07}, the ascending order of processing times of jobs in $A\setminus\{j_1\}$ provides the longest schedule on the shared processor.
Thus, in other words, for processing-time-inclusive jobs $A$, the makespan $x$ is shorter than the shortest job in $A$.
The condition can be checked in time $O(|\jobs| \log |\jobs|)$.

Finally, for a permutation $j_1,\ldots,j_k$ of jobs with weights $w_1,\ldots,w_k$, respectively,
define 
\begin{equation} \label{W}
W_i=\sum_{\ell=i+2}^{k} \frac{w_{\ell}}{2^{\ell-i-1}}
\end{equation}
for each $i\in\{-1,0,1,\ldots,k-2\}$. 
\begin{lemma} \label{lem:switching}
Let $\cS$ be a synchronized schedule that executes processing-time-inclusive jobs $j_1,\ldots,j_k$ with processing times $p_1,\ldots,p_k$ and weights $w_1,\ldots,w_k$, respectively,  in the order $j_1,\ldots,j_k$  on a shared processor $\cM$, and let $i\in\{1,\ldots,k-1\}$.
Let $\cS'$ be a synchronized schedule obtained by exchanging jobs $j_i$ and $j_{i+1}$ in $\cS$ so that all jobs are executed in the order $j_1,\ldots,j_{i-1},j_{i+1},j_i,j_{i+2},\ldots,j_k$ on $\cM$.
Then,
\[\tct{\cS}-\tct{\cS'}= \frac{(w_{i+1}-w_i)T_i}{4} + \frac{(p_i-p_{i+1})W_{i+1}}{4} + \frac{w_ip_{i+1}}{4} - \frac{w_{i+1}p_i}{4}.\]
\end{lemma}
\begin{proof}
Note that the construction of $\cS'$ is valid since the jobs are processing-time-inclusive.
We calculate the values of $\cS$ and $\cS'$ on $\cM$ only since the schedules on other shared processors remain unchanged and thus contribute the same amount $\sigma$ to the total weighted overlap of both schedules.
By Observation~\ref{obs:T_i} we have
\begin{equation} \label{eq:S:sum}
\tct{\cS} = \sigma + \sum_{\ell=1}^k \bar{t_{\ell}}w_{\ell},
\end{equation}
for $\cS$ and  
\begin{equation} \label{eq:S':sum}
\tct{\cS'} = \sigma + \sum_{\ell=1}^{i-1} \bar{t_{\ell}}w_{\ell}+ \frac{(p_{i+1}-T_i)w_{i+1}}{2} + \frac{(p_{i}-p_{i+1}/2-T_i/2)w_{i}}{2} + \sum_{\ell=i+2}^{k} (T_{\ell+1}'-T_{\ell}')w_{\ell},
\end{equation}
for $\cS'$, where
\[T_{\ell}'=T_i+\frac{p_{i+1}}{2^{\ell-i}}+\frac{p_i}{2^{\ell-i-1}}+\sum_{\ell'=i+2}^{\ell-1}\frac{p_{\ell'}}{2^{\ell-\ell'}}\]
for each $\ell\in\{i+2,\ldots,k+1\}$.
We obtain 
\[T_{\ell}-T_{\ell}'=\frac{p_i}{2^{\ell-i}} + \frac{p_{i+1}}{2^{\ell-i-1}} - \frac{p_{i+1}}{2^{\ell-i}} - \frac{p_i}{2^{\ell-i-1}} = \frac{p_{i+1}-p_i}{2^{\ell-i}},\]
for each $\ell\in\{i+2,\ldots,k+1\}$.
Thus,
\begin{equation} \label{eq:W_i}
\left(\sum_{\ell=i+2}^k (T_{\ell+1}-T_{\ell})w_{\ell}-\sum_{\ell=i+2}^{k} (T_{\ell+1}'-T_{\ell}')w_{\ell}\right) = (p_i-p_{i+1})\sum_{\ell=i+2}^{k}\frac{w_{\ell}}{2^{\ell+1-i}}.
\end{equation}
For notational brevity set 
\begin{equation} \label{eq:fw_i}
f(w_i) := w_i\left( T_{i+1}-T_i+\frac{T_i}{4}-\frac{p_i}{2}+\frac{p_{i+1}}{4}\right) = w_i\left(\frac{T_i}{2}+\frac{p_i}{2}-\frac{3T_i}{4}-\frac{p_i}{2}+\frac{p_{i+1}}{4} \right) = w_i\left(\frac{p_{i+1}}{4}-\frac{T_i}{4}\right),
\end{equation}
and 
\begin{equation} \label{eq:fw_i+1}
f(w_{i+1}) := w_{i+1}\left( T_{i+2}-T_{i+1}-\frac{p_{i+1}}{2}+\frac{T_i}{2} \right) = w_{i+1}\left( \frac{T_i}{4}+\frac{p_{i+1}}{2}+\frac{p_i}{4} -\frac{p_i}{2}-\frac{p_{i+1}}{2} \right) = w_{i+1}\left(\frac{T_i}{4}-\frac{p_i}{4}\right).
\end{equation}
By \eqref{eq:S:sum}, \eqref{eq:S':sum}, \eqref{eq:W_i}, \eqref{eq:fw_i} and \eqref{eq:fw_i+1} we obtain
\begin{eqnarray*} 
\tct{\cS}-\tct{\cS'} & = & \left((T_{i+1}-T_{i})w_{i}-\frac{(p_{i+1}-T_i)w_{i+1}}{2}\right) + \left((T_{i+2}-T_{i+1})w_{i+1}-\frac{(p_{i}-p_{i+1}/2-T_i/2)w_{i}}{2}\right) \\
                     & + & \left(\sum_{\ell=i+2}^k (T_{\ell+1}-T_{\ell})w_{\ell}-\sum_{\ell=i+2}^{k} (T_{\ell+1}'-T_{\ell}')w_{\ell}\right) \\
                     & = & f(w_i)+f(w_{i+1}) + (p_i-p_{i+1})\sum_{\ell=i+2}^{k} \frac{w_{\ell}}{2^{\ell+1-i}} \\
                     & = & \frac{(w_{i+1}-w_i)T_i}{4} + \frac{(p_i-p_{i+1})W_{i+1}}{4} + \frac{w_ip_{i+1}}{4} - \frac{w_{i+1}p_i}{4}
\end{eqnarray*}
as required.
\end{proof}

The next observation that follows directly from Lemma~\ref{lem:pulling} (see also  (\ref{eq:synchronized:ineq})).
\begin{observation} \label{obs:wW}
Let $\cS$ be an optimal synchronized schedule that executes jobs $j_1,\ldots,j_k$ with 
weights $w_1,\ldots,w_k$, respectively,
on a shared processor in the order $j_1,\ldots,j_k$.
Then, $w_i\geq W_{i-1}$ for each $i\in\{1,\ldots,k-1\}$.
\qed
\end{observation}

We finish this section with the following feature of optimal synchronized schedules.
\begin{observation} \label{obs:TW}
Let $\cS$ be an optimal synchronized schedule that executes 
jobs $j_1,\ldots,j_k$ with processing times $p_1,\ldots,p_k$ and weights $w_1,\ldots,w_k$, respectively, on a shared processor in the order $j_1,\ldots,j_k$.
Then,
\[0=T_1< T_2 < \cdots < T_{k+1},\]
and
\[W_{-1}\geq W_0\geq W_1\geq\cdots\geq W_{k-2}.\]
\end{observation}
\begin{proof}
By Observation \ref {obs:T_i} we have $T_{i+1}=\frac{p_i}{2}+\frac{T_i}{2}$ for each $i\in\{1,\ldots,k\}$. Since $j_i$ is executed on the shared processor, we have $p_i>T_i$.
Thus, $T_{i+1}>\frac{T_i}{2}+\frac{T_i}{2}=T_i$ for each $i\in\{1,\ldots,k\}$.

By~\eqref{W} we have $W_i=\frac{w_{i+2}}{2}+\frac{W_{i+1}}{2}$ for each $i\in\{-1,0,1,\ldots,k-3\}$.
By Observation \ref{obs:wW}, $w_{i+2}\geq W_{i+1}$ for each $i\in\{-1,0,1,\ldots,k-3\}$.
Thus, $W_i\geq \frac{W_{i+1}}{2}+\frac{W_{i+1}}{2}=W_{i+1}$ for each $i\in\{-1,0,1,\ldots,k-3\}$.
\end{proof}

\subsubsection{Instances with $p_i=w_i$} \label{lem:V:restricted}

Let $\cS$ be a synchronized schedule that executes jobs $j_1,\ldots,j_k$ with processing times $p_1,\ldots,p_k$ on a shared processor in the order $j_1,\ldots,j_k$. The schedule $\cS$ is called \emph{V-shaped} if, for each shared processor, there exists an index $\ell\in\{1,\ldots,k\}$ such that $p_1\geq p_2\geq\cdots\geq p_{\ell}$ and $p_{\ell}\leq p_{\ell+1}\leq\cdots\leq p_{k}$.

\begin{lemma} \label{lem:Vshape}
Optimal synchronized schedules for instances with processing-time-inclusive jobs $\jobs$ and with  $p_i=w_i$ for $i\in \jobs$,
are V-shaped.
\end{lemma}
\begin{proof}
Let $\cS$ be an optimal synchronized schedule for $\jobs$.
Take an arbitrary shared processor, and let $j_1,\ldots,j_k$ be the order of jobs on this processor.
Since $\cS$ is optimal, by Lemma~\ref{lem:switching} for each $i\in\{1,\ldots,k-1\}$ (note that the jobs in $\jobs$ are processing-time-inclusive by assumption as required in the lemma),
\begin{equation} \label{eq:non-positive}
\frac{(p_{i+1}-p_i)(T_i-W_{i+1})}{4} \geq 0
\end{equation}
because $w_i=p_i$ and $w_{i+1}=p_{i+1}$.
Denote $b_i=(p_{i+1}-p_i)/|p_{i+1}-p_i|$ for each $i\in\{1,\ldots,k-1\}$. By definition $b_i=0$ for $p_{i+1}=p_i$.
Note that if all $b_i$'s are non-negative or all of them are non-positive, then $p_1\leq\cdots\leq p_k$ and $p_1\geq\cdots\geq p_k$, respectively, and hence $\cS$ is V-shaped.
Define $l=\max\{i\st b_i=-1\}$ and $r=\min\{i\st b_i=1\}$.
Note that $l<r$ implies
\[ p_1\geq p_2\geq\cdots\geq p_l \quad\textup{and}\quad p_{l}\leq p_{l+1}\leq\cdots\leq p_k\]
and thus $\cS$ is V-shaped.
Hence, it remains to argue that $l<r$.
Suppose for a contradiction that $l>r$ (note that $l\neq r$ by definition).
By \eqref{eq:non-positive},
\[\frac{b_r(T_r-W_{r+1})}{4} \geq 0\]
which implies that $T_r\geq W_{r+1}$.
By Observation~\ref{obs:TW}, $W_{r+1}\geq W_{l+1}$ and $T_r<T_l$, which implies $T_l> W_{l+1}$.
Since $b_l=-1$, this gives $b_l(T_l-W_{l+1})<0$.
Hence,
\[\frac{(p_{l+1}-p_l)(T_l-W_{l+1})}{4} < 0\]
which contradicts \eqref{eq:non-positive} and completes the proof.
\end{proof}

\subsection{Duality and Reversibility}

This section introduces duality of processing times and weights in the WSMP problem. This duality will be used in  the next section to prove the problem strong NP-hardness.
The duality is particularly easy to observe from a matrix representation of the total weighted overlap of a synchronized schedule.
Let us define two $k\times k$ matrices
\[
\vectL_k = \begin{bmatrix}
            0          & 0          & 0          & \cdots & 0        & 0 \\
            2^{-1}     & 0          & 0          & \cdots & 0        & 0 \\
            2^{-2}     & 2^{-1}     & 0          & \cdots & 0        & 0 \\
            \vdots     & \vdots  & \vdots  & \ddots & \vdots   & \vdots \\
            2^{-(k-2)} & 2^{-(k-3)} & 2^{-(k-2)} & \cdots & 0        & 0 \\
            2^{-(k-1)} & 2^{-(k-2)} & 2^{-(k-3)} & \cdots & 2^{-1}   & 0 \\
           \end{bmatrix},
\quad \textup{and}\quad
\vectU_k= \begin{bmatrix}
            0       & 2^{-1}     & \cdots & 2^{-(k-3)}  & 2^{-(k-2)} & 2^{-(k-1)} \\
            0       & 0          & \cdots & 2^{-(k-4)}  & 2^{-(k-3)} & 2^{-(k-2)} \\
            0       & 0          & \cdots & 2^{-(k-5)}  & 2^{-(k-4)} & 2^{-(k-3)}\\
            \vdots  & \vdots     & \ddots & \vdots      & \vdots     & \vdots \\
            0       & 0          & \cdots & 0           & 0          & 2^{-1} \\
            0       & 0          & \cdots & 0           & 0          & 0 \\
           \end{bmatrix}.
\]
Let $\vectW$ be the vector of weights, $\vectW=[w_1\,\ldots\,w_k]$, and $\vectP$ be the vector of processing times, $\vectP=[p_1\,\ldots\,p_k]$.
Since $(\vectA\cdot\vectB)^{\T}=\vectB^{\T}\cdot\vectA^{\T}$, we observe the following.
\begin{observation} \label{obs:matrices}
It holds $\vectW\cdot\vectL_k\cdot\vectP^{\T}=\vectP\cdot\vectU_k\cdot\vectW^{\T}$.
\qed
\end{observation}
The above matrix notation can be conveniently used to express the total weighted overlap of a given schedule as stated in the next lemma.
\begin{lemma} \label{lem:SviaMatrices}
Let $\cS$ be a synchronized schedule that executes jobs $j_1,\ldots,j_k$ with processing times $p_1,\ldots,p_k$ and weights $w_1,\ldots,w_k$, respectively, on a single shared processor in the order $j_1,\ldots,j_k$.
Then, 
\[\tct{\cS}=\frac{1}{2}\vectP\cdot\vectI_k\cdot\vectW^{\T}-\frac{1}{2}\vectW\cdot\vectL_k\cdot\vectP^{\T}=\frac{1}{2}\vectW\cdot\vectI_k\cdot\vectP^{\T}-\frac{1}{2}\vectP\cdot\vectU_k\cdot\vectW^{\T},\]
 where $\vectI_k$ is the $k\times k$ identity matrix.
\end{lemma}
\begin{proof}
By Observation \ref{obs:T_i} we have
\[
\tct{\cS}  =  \sum_{\ell=1}^k (T_{\ell+1}-T_{\ell})w_{\ell} 
           =  \sum_{\ell=1}^k \left(\sum_{t=1}^{\ell}\frac{p_t}{2^{\ell+1-t}} - \sum_{t=1}^{\ell-1}\frac{p_t}{2^{\ell-t}}\right) w_{\ell} 
           =  \sum_{\ell=1}^k \frac{p_{\ell}w_{\ell}}{2} - \sum_{\ell=1}^k \sum_{t=1}^{\ell-1} \frac{p_t w_{\ell}}{2^{\ell+1-t}}.
\]
Note that
\[\sum_{\ell=1}^k \frac{p_{\ell}w_{\ell}}{2} = \frac{1}{2}\vectP\cdot\vectI_k\cdot\vectW^{\T}=\frac{1}{2}\vectW\cdot\vectI_k\cdot\vectP^{\T}\]
and, by Observation~\ref{obs:matrices},
\[\sum_{\ell=1}^k \sum_{t=1}^{\ell-1} \frac{p_t w_{\ell}}{2^{\ell+1-t}} = \frac{1}{2}\vectW\cdot\vectL_k\cdot\vectP^{\T}= \frac{1}{2}\vectP\cdot\vectU_k\cdot\vectW^{\T},\]
which completes the proof.
\end{proof}

This lemma points out at a duality of processing times and weights in the WSMP  problem, namely, $\frac{1}{2}\vectW\cdot\vectI_k\cdot\vectP^{\T}-\frac{1}{2}\vectP\cdot\vectU_k\cdot\vectW^{\T}$ is a transposition
of $\frac{1}{2}\vectP\cdot\vectI_k\cdot\vectW^{\T}-\frac{1}{2}\vectW\cdot\vectL_k\cdot\vectP^{\T}$. The former takes the weights for processing times and the processing times for the weights from the latter, and the $\vectU_k=\vectL_k^{\T}$ reverses the order of jobs from $1,\ldots,k$ to $k,\ldots,1$. Unfortunately, the reversed order may not result in a feasible schedule on the shared processor in general since  it may no longer be possible to execute some jobs on the shared processor according to that order because the jobs may prove too short for that. Recall that we exchanged processing times for weights and vice versa besides reversing the order.
Again, the test whether this actually happens depends not only on the weights of  jobs but also on their \emph{order}. Therefore we introduce
an easy to test sufficient condition that would always guarantee the validity of the reversed order.

We now introduce a concept analogous to the one of processing-time-inclusive jobs but for the weights.
Consider a set of jobs $A=\{j_1,\ldots,j_k\}$ such that $w_{1}\leq w_{2}\leq\cdots\leq w_{k}$, where $w_i$ is the weight of $j_i$, $i\in\{1,\ldots,k\}$.
We say that the set of jobs $A$ is \emph{weight-inclusive} if
\[x=\sum_{\ell=1}^{|A|-1} \frac{w_{\ell+1}}{2^{|A|-\ell}} < w_1.\]
Note that $x$ is the makespan of a schedule on a shared processor for jobs in $A\setminus\{j_1\}$ when the jobs are scheduled in ascending order of their weights, i.e., the order $j_2,\ldots,j_k$, and the processing time of $j_i$ equals its weight $w_i$ for each $i\in\{2,\ldots,k\}$.
Again, by~\cite{VairaktarakisAydinliyim07}, this order of jobs in $A\setminus\{j_1\}$ provides the longest schedule on the shared processor.
The condition can be checked in time $O(|\jobs| \log |\jobs|)$.
We remark that if $A'$ is a set of $k$ jobs such that the processing time of the $i$-th job in $A'$ equals $w_i$, then $A$ is weight-inclusive if and only if $A'$ is processing-time-inclusive.

We have the following duality lemma.
\begin{lemma} \label{lem:reversability}
Let $\mathcal{M}$ be a shared processor with jobs $\jobs$. Suppose that $\jobs$ is both processing-time-inclusive and weight-inclusive.
Let $\cS$ be any synchronized schedule and let $\cS'$ be a synchronized schedule obtained from $\cS$ by reversing the order of jobs on $\mathcal{M}$, and by exchanging the processing times for weights and the weights for processing times.
Then, $\tct{\cS}=\tct{\cS'}$.
\qed
\end{lemma}
Observe that for the case $p_i=w_i$, the processing-time-inclusion for $\jobs$ on $\mathcal{M}_{\ell}$ implies the weight-inclusion for $\jobs$ on $\mathcal{M}_{\ell}$, and the duality reduces to a schedule reversibility.

\section{WSMP is NP-hard in the strong sense} \label{sec:NPC}

In this section we prove, by a transformation from the Numerical 3-Dimensional Matching (N3DM) \cite{GareyJohnson79}, that the decision version of weighted multiple shared-processors (WMSM) problem is strongly NP-hard even if for each job its processing time and weight are equal.
The N3DM problem input consists of three multisets of integers $X =\{x_{1} ,\ldots  ,x_{n}\} ,Y =\{y_{1} ,\ldots  ,y_{n}\}$ and $Z =\{z_{1} ,\ldots  ,z_{n}\}$, and an integer $b$.
The decision question is: does there exist multisets $S_1,\ldots,S_n$, each of size $3$, such that $\bigcup_{i=1}^n S_i=X\cup Y\cup Z$ and for each $i\in\{1,\ldots,n\}$ it holds $\sum_{a\in S_i}a=b$, $|X\cap S_i|=1$, $|Y\cap S_i|=1$ and $|Z\cap S_i|=1$?
In this section we use $\xi(\mathbf{A})$ to denote the sum of all entries of a matrix $\mathbf{A}$.

We construct an instance of the WMSM problem as follows.
The weights are equal to processing times for all jobs.
There are $3 n$ jobs and $n$ shared processors.
The jobs are split into three sets $A ,B$ and $C$ of equal size $n$.
The jobs in $A$ have processing times
\begin{equation*}
s_{i} =2 (M +m +x_{i}) =2 a_{i},
\end{equation*}
the jobs in $B$ have processing times
\begin{equation*}
b_{i} =2 M +y_{i},
\end{equation*}
and the jobs in $C$ have processing times
\begin{equation*}
r_{i} =2 (M +m^{2} +z_{i}) =2 c_{i}.
\end{equation*}
We take the integers $M$ and $m$ as follows:
\begin{equation} \label{eq:Mm}
M >7(m^2+b) \quad \textup{ and }\quad m >\max \{b ,6\}.
\end{equation}
Informally speaking, the $M$ is to guarantee that each shared processor has exactly three jobs in an optimal schedule, the $m$ is to guarantee that each shared processor does exactly one job from each of the sets $A$, $B$ and $C$.

A synchronized schedule $\cS$ for the above instance is called \emph{equitable} if each shared processor executes exactly three jobs $i\in A$, $j\in B$ and $k\in C$ with the ordering $(i, j, k)$.

For brevity we define:
\[h(\Delta_1,\ldots,\Delta_n) := \sum _{l =1}^{n}\left(\frac{15}{8}a_{l}^2+\frac{3}{8}b_{l}^2 + \frac{15}{8}c_{l}^2\right) -\frac{1}{4} \sum_{l=1}^n\left(4 M +m +m^{2} +b-\Delta_l\right)^{2}\]
for any integers $\Delta_1,\ldots,\Delta_n$.
The lower bound in the decision counterpart of the WMSM is set to $h(0,\ldots,0)$.

The outline of the proof is as follows.
In Lemma~\ref{lem:total} we provide a formula for the total weighted overlap of a given equitable schedule.
Informally speaking, this lemma provides in particular a one-to-one correspondence between the total weighted overlaps of equitable schedules and the values of $h(\Delta_1,\ldots,\Delta_n)$.
This reduces the task of finding equitable schedules that maximize the total weighted overlap to finding values of $\Delta_1,\ldots,\Delta_n$ that maximize the function $h$.
These values are $\Delta_1=\cdots=\Delta_n=0$ as indicated above.  We use this observation in Lemma~\ref{lem:NPCiff}; this key lemma proves the correspondence between N3DM and WMSM but it works with equitable schedules only.
More precisely, we argue in Lemma~\ref{lem:NPCiff} that there exists a solution to the N3DM problem if and only if there exists an equitable schedule $\cS$ for the WMSM problem for which it holds $\tct{\cS}\geq h(0,\ldots,0)$.
Finally, Lemma~\ref{lem:equitable} justifies restricting attention to equitable schedules only: each optimal schedule for an instance of WMSM constructed from an input to N3DM problem is equitable.
Thus, these three lemmas prove our NP-hardness result stated in Theorem~\ref{thm:NPC}.

\begin{lemma} \label{lem:total}
For an equitable schedule $\cS$ it holds
\[\tct{\cS}= h(\Delta_1,\ldots,\Delta_n), \]
where $\Delta_{l} =b -(x_{i} +y_{j} +z_{k})$ and $i, j, k$ are jobs from $A ,B$ and $C$, respectively, done on shared processor $\Mshared_l$.
\end{lemma}
\begin{proof}
Consider three jobs $i \in A ,j \in B ,k \in C$ scheduled with the ordering $(i, j, k)$ on some shared processor $\Mshared_l$. Denote this schedule by $\cS_l$. Let $\vectP_l=[s_i, b_j, r_k]$ be the vector of processing times of jobs $i, j, k$.
Since processing time equals the weight for each job, by Lemma~\ref{lem:SviaMatrices}
we have $\tct{\cS_l}=\frac{1}{2}\vectP_l\cdot\vectI_3\cdot\vectP_l^{\T}-\frac{1}{2}\vectP_l\cdot\vectL_3\cdot\vectP_l^{\T}=\frac{3}{4}\vectP_l\cdot\vectI_3\cdot\vectP_l^{\T}-\frac{1}{2}\xi(\mathbf{A}_l)$ where
\[
\mathbf{A}_l=\begin{bmatrix}
\frac{1}{2} s_{i} s_{i} & \frac{1}{4} s_{i} b_{j} & \frac{1}{8} s_{i} r_{k} \\
\frac{1}{4} b_{j} s_{i}  & \frac{1}{2} b_{j} b_{j} & \frac{1}{4} b_{j} r_{k} \\
\frac{1}{8}  r_{k} s_{i} & \frac{1}{4} r_{k} b_{j} & \frac{1}{2} r_{k} r_{k}
\end{bmatrix}
=
\begin{bmatrix}
2 a_{i} a_{i} & \frac{1}{2} a_{i} b_{j} & \frac{1}{2} a_{i} c_{k} \\
\frac{1}{2}  b_{j} a_{i} & \frac{1}{2} b_{j} b_{j} & \frac{1}{2} b_{j} c_{k} \\
\frac{1}{2} c_{k}  a_{i} & \frac{1}{2} c_{k} b_{j} & 2 c_{k} c_{k}
\end{bmatrix}
= \mathbf{B}_l + \mathbf{C}_l,\]
and where
\[
\mathbf{B}_l=\begin{bmatrix}
\frac{1}{2} a_{i} a_{i} & \frac{1}{2} a_{i} b_{j} & \frac{1}{2} a_{i} c_{k} \\
\frac{1}{2} b_{j}  a_{i} & \frac{1}{2} b_{j} b_{j} & \frac{1}{2} b_{j} c_{k} \\
\frac{1}{2} c_{k}  a_{i} & \frac{1}{2} c_{k} b_{j} & \frac{1}{2} c_{k} c_{k}
\end{bmatrix}
\textup{ and }
\mathbf{C}_l=\begin{bmatrix}
\frac{3}{2} a_{i} a_{i} & 0 & 0 \\
0 & 0 & 0 \\
0 & 0 & \frac{3}{2} c_{k} c_{k}
\end{bmatrix}.
\]
We have
\begin{equation*}
\xi(\mathbf{B}_l)=\frac{1}{2} (a_{i} +b_{j} +c_{k})^{2} =\frac{1}{2} (M +m +x_{i} +2 M +y_{j} +M +m^{2} +z_{k})^{2} =\frac{1}{2} (4 M +m +m^{2} +x_{i} +y_{j} +z_{k})^{2}
\end{equation*}
and
\begin{equation*}
\xi(\mathbf{C}_l) = \frac{3}{2} a_{i}^{2} +\frac{3}{2} c_{k}^{2}.
\end{equation*}
Therefore,
\begin{equation} \label{eq:total}
\sum_{l=1}^n\xi(\mathbf{A}_l) = \frac{1}{2}\sum_{l=1}^n (4 M +m +m^{2} +b - \Delta_l)^{2} +\frac{3}{2}\sum _{l =1}^{n}(a_{l}^{2} +c_{l}^{2}).
\end{equation}
We finally obtain:
\begin{eqnarray*}
\tct{\cS} & = & \sum_{l=1}^n \tct{\cS_l} =  \sum_{l=1}^n \frac{3}{4}\vectP_l\cdot\vectI_3\cdot\vectP_l^{\T} - \frac{1}{2}\sum_{l=1}^n \xi(\mathbf{A}_l) \\
          & = & \frac{3}{4}\sum_{l=1}^n \left( \frac{1}{2}s_{l}^2+\frac{1}{2}b_{l}^2 + \frac{1}{2}r_{l}^2\right) - \frac{1}{2}\sum_{l=1}^n \xi(\mathbf{A}_l) \\
          & = & \sum_{l=1}^n \left( \frac{15}{8}a_{l}^2+\frac{3}{8}b_{l}^2 + \frac{15}{8}c_{l}^2\right) - \frac{1}{4}\sum_{l=1}^n (4 M +m +m^{2} +b - \Delta_l)^{2} \\
          & = & h(\Delta_1,\ldots,\Delta_n).
\end{eqnarray*}
\end{proof}

\begin{lemma} \label{lem:NPCiff}
There exists a solution to the \textup{N3DM} problem with the input $X,Y,Z$ and $b$ if and only if for the set of jobs $A\cup B\cup C$ and $n$ shared processors there exists an equitable schedule $\cS$ such that $\tct{\cS} \geq h(0,\ldots,0)$.
\end{lemma}
\begin{proof}
($\Longrightarrow$)
Suppose that there exists a solution $S_1,\ldots,S_n$ to the \textup{N3DM} problem, where take for convenience $S_i=\{x_i,y_i,z_i\}$ for each $i\in\{1,\ldots,n\}$.
Construct a schedule $\cS$ such that the $i$-th shared processor executes the jobs with processing times $s_i,b_i,r_i$ in this order.
Since $\cS$ is equitable, Lemma~\ref{lem:total} implies that $\tct{\cS}=h(\Delta_1,\ldots,\Delta_n)$, where $\Delta_i=b-(x_1+y_i+z_i)$ for each $i\in\{1,\ldots,n\}$.
Since $S_1,\ldots,S_n$ is a solution to the \textup{N3DM} problem, $x_i+y_i+z_i=b$ for each $i\in\{1,\ldots,n\}$ and therefore $\tct{\cS}=h(0,\ldots,0)$ as required.

\medskip\noindent
($\Longleftarrow$)
Suppose there is an equitable schedule $\cS$ on $n$ shared processors with $\tct{\cS}\geq h(0,\ldots,0)$.
Recall that by definition of equitable schedule each shared processor does exactly three jobs, the first one from  $A$, the second from $B$ and the third from $C$.
By Lemma~\ref{lem:total}, $\tct{\cS}= h(\Delta_1,\ldots,\Delta_n)$, where $\Delta_{l} =b -(x_{l} +y_{l} +z_{l})$ and $s_{l} ,b_{l} ,r_{l}$ are the processing times of jobs from $A ,B$ and $C$, respectively, done on the $l$-th shared processor for each $l\in\{1,\ldots,n\}$.
Denote
\[g(\Delta_1,\ldots,\Delta_n) := \sum_{l=1}^n\left(4 M +m +m^{2} +b-\Delta_l\right)^{2}.\]
Since\begin{equation*}
\sum _{l =1}^{n} \Delta _{l} = 0
\end{equation*}
we have
\begin{equation*}
g(\Delta_1,\ldots,\Delta_n) = n (4 M +m +m^{2} +b)^{2} +\sum _{l=1}^{n} \Delta_{l}^{2}.
\end{equation*}
By definition,
\[h(\Delta_1,\ldots,\Delta_n)\geq h(0,\ldots,0) \quad\Leftrightarrow\quad g(\Delta_1,\ldots,\Delta_n)\leq g(0,\ldots,0).\]
Moreover,
\[g(\Delta_1,\ldots,\Delta_n)\leq g(0,\ldots,0) \quad\Leftrightarrow\quad \sum_{l=1}^{n}\Delta_{l}^2\leq 0.\]
Thus, $\Delta_{l} =0$ for each $l\in\{1,\ldots,n\}$ and hence $x_l+y_l+z_l=b$ for each $l\in\{1,\ldots,n\}$ which implies that $X,Y,Z$ and $b$ is a solution to N3DM.
\end{proof}

It remains to justify our earlier assumption that it is sufficient to limit ourselves to equitable schedules only.
\begin{lemma} \label{lem:equitable}
For the instance $\jobs=A\cup B\cup C$ on $n$ shared processors constructed from the input to the \textup{N3DM} problem, each optimal schedule is  equitable.
\end{lemma}
\begin{proof}
We first prove that each shared processor does exactly three jobs in any optimal schedule.
Suppose for a contradiction that this is not the case in some optimal schedule $\cS$.
Then there exist shared processors $\Mshared_{l'}$ and $\Mshared_{l}$ that execute $x'<3$ and $x>3$ jobs, respectively.
We obtain a new schedule $\cS'$ from $\cS$ by moving a job $j$ from the last position $x$ on $\Mshared_{l}$ to the last position $x'+1$ on $\Mshared_{l'}$.
Observe that the processing time of each job is at most $2(M+m^2+b)$ and thus the last job on $\Mshared_{l'}$ completes in $\cS$ by $\frac{3}{2}(M+m^2+b)$ and the shortest job in the instance is not shorter than $2M$, thus $\frac{3}{2}(M+m^2+b)<2M$ for $M>7(m^2+b)$ (as guaranteed by \eqref{eq:Mm}) and consequently the job $j$ is long enough to be executed in position $x'+1$ on $\Mshared_{l'}$.
Since the transition from $\cS$ to $\cS'$ does not affect the execution intervals of any job except for $j$, we obtain by Observation~\ref{obs:T_i}
\begin{equation} \label{eq:all3-1}
\tct{\cS}-\tct{\cS'} = w_j\left(\frac{T_{x'+1}'}{2}-\frac{T_x}{2}\right),
\end{equation}
where $T_{x'+1}'$ and $T_x$ are completion times of the last job in $\cS$ on processors $\Mshared_{l'}$ and $\Mshared_l$, respectively.
The maximum job processing time in $A\cup B\cup C$ does not exceed $2(M+m^2+b)$ and hence by Observation~\ref{obs:T_i}
\[
T_{x'+1}'\leq \sum_{\ell=1}^{x'}\frac{2(M+m^2+b)}{2^{x'+1-\ell}} = 2(M+m^2+b)\left(1-\frac{1}{2^{x'}}\right) \]
and the minimum job processing time of a job in $A\cup B\cup C$ is not less than $2M$ which gives
\[
T_x \geq \sum_{\ell=1}^{x-1}\frac{2M}{2^{x-\ell}} = 2M\left(1-\frac{1}{2^{x-1}}\right).\]
Thus, since $x'<3$ and $x'< x-1$
\[T_{x'+1}'-T_x \leq 2M\left(\frac{1}{2^{x-1}}-\frac{1}{2^{x'}}\right)+ 2(m^2+b)\left(1-\frac{1}{2^{x'}}\right)<
-2M\frac{1}{2^{x'+1}}+
2(m^2+b)\left(1-\frac{1}{2^3}\right).
\]
However,
\[-M \frac{1}{2^{x'+1}}+(m^2+b)
\left(1-\frac{1}{2^3}\right)<0
\]
for $M>7(m^2+b)$ (as guaranteed by \eqref{eq:Mm})
which gives $\tct{\cS}-\tct{\cS'}<0$ and contradicts the optimality of $\cS$.
Thus, we have proved that each shared processor executes exactly three jobs in each optimal schedule.

In the following we will often compare lengths of jobs from the sets $A$, $B$ and $C$.
In particular, by~\eqref{eq:Mm}, we have that for each $i,j,k\in\{1,\ldots,n\}$,
\begin{equation} \label{eq:ABC}
b_i \leq 2M+b \leq 2(M+m) < s_j < 2(M+m+b) \leq 2(M+m^2) < r_k.
\end{equation}
Informally, each job in $C$ is longer than any job in $A$, and each job in $A$ is longer than any job in $B$.

\medskip
We now prove that each shared processor does exactly one job from $C$.
Consider an optimal schedule $\cS$ in which some shared processor $\Mshared_l$ executes at least two jobs $i$ and $k$ from $C$.
Then, no jobs from $C$ are on another shared processor $\Mshared_{l'}$.
Without loss of generality we may assume due to~\eqref{eq:ABC} that $i$ and $k$ are the longest and the second longest jobs respectively on $\Mshared_l$.
Denote by $j$ the third job on $\Mshared_l$.
By Lemma~\ref {lem:Vshape}, an optimal schedule on $\Mshared_l$ is V-shaped.
Thus, the order of jobs on $\Mshared_l$ is either $(i,k,j)$, $(j,k,i)$, $(i,j,k)$ or $(k,j,i)$.
By Lemma~\ref{lem:reversability}, we can further reduce the number of cases to $(i, k, j)$ and $(i, j, k)$.
It can be easily checked, we omit details here, that the former order is not optimal  on $\Mshared_l$ since $r_k\geq q_j$, where $q_j$ is the processing time of the job $j$.
Thus, it suffices to consider the order $(i, j, k)$ on $\Mshared_l$.
Let $q_{i'},q_{j'},q_{k'}$ be the processing times of jobs scheduled with the ordering $(i',j',k')$ on $\Mshared_{l'}$.
By Lemma~\ref{lem:SviaMatrices}, we have $\tct{\cS}=\sigma-\xi(\mathbf{A})/2-\xi(\mathbf{A}')/2$, where
\[\mathbf{A}=\begin{bmatrix}
0 & \frac{1}{4} r_{i} q_{j} & \frac{1}{8} r_{i} r_{k} \\
\frac{1}{4}  q_{j} r_{i} & 0 & \frac{1}{4} q_{j} r_{k} \\
\frac{1}{8}  r_{k} r_{i} & \frac{1}{4} r_{k} q_{j} & 0
\end{bmatrix},
\quad
\mathbf{A}'=\begin{bmatrix}
0 & \frac{1}{4} q_{i'} q_{j'} & \frac{1}{8} q_{i'} q_{k'} \\
\frac{1}{4}  q_{j'} q_{i'} & 0 & \frac{1}{4} q_{j'} q_{k'} \\
\frac{1}{8}  q_{k'} q_{i'} & \frac{1}{4} q_{k'} q_{j'} & 0
\end{bmatrix}.
\]
and  $\sigma=\sum_{i\neq l, l'} \tct{\cS_i}+\frac{1}{2}\vectP_l\cdot\vectI_3\cdot\vectP_l^{\T}+\frac{1}{2}\vectP_{l'}\cdot\vectI_3\cdot\vectP_{l'}^{\T}$
in which we take $\cS_i$ to be the schedule on $\Mshared_i$ and the corresponding private processors assigned to jobs executed on $\Mshared_i$.
Consider the matrices
\[\mathbf{B}=\begin{bmatrix}
0 & \frac{1}{4} q_{i'}  q_{j} & \frac{1}{8} q_{i'} r_{k} \\
\frac{1}{4}  q_{j} q_{i'} & 0 & \frac{1}{4} q_{j} r_{k} \\
\frac{1}{8}  r_{k} q_{i'} & \frac{1}{4} r_{k} q_{j} & 0
\end{bmatrix},
\quad
\mathbf{B}'=\begin{bmatrix}
0 & \frac{1}{4} r_{i} q_{j'} & \frac{1}{8} r_{i} q_{k'} \\
\frac{1}{4} q_{j'} r_{i}  & 0 & \frac{1}{4} q_{j'} q_{k'} \\
\frac{1}{8} q_{k'} r_{i}  & \frac{1}{4} q_{k'} q_{j'} & 0
\end{bmatrix}.
\]
obtained from $\mathbf{A}$ and $\mathbf{A}'$, respectively, by exchanging $r_i$ and $q_{i'}$.
Thus, there exists a schedule $\cS'$ obtained from $\cS$ by exchanging  job $i$ on $\Mshared_l$ with job $i'$ on $\Mshared_{l'}$, $\tct{\cS'}=\sigma-\xi(\mathbf{B})/2-\xi(\mathbf{B}')/2$.
Observe that by~\eqref{eq:ABC}, in $\cS'$, the first jobs on $\Mshared_{l}$ and $\Mshared_{l'}$ complete by $(M+m^2+b)$, the second jobs on those processors  completes by $\frac{3}{2}(M+m^2+b)$, and moreover the shortest job in the instance is not shorter than $2M$, thus $\frac{3}{2}(M+m^2+b)<2M$ for $M>7(m^2+b)$ as guaranteed by~\eqref{eq:Mm} and consequently all jobs on $\Mshared_{l}$ and $\Mshared_{l'}$ are  long enough to be executed on $\Mshared_{l}$ and $\Mshared_{l'}$ after the exchange.
Therefore, $\cS'$ is feasible.
We have
\begin{equation*}
\tct{\cS'}-\tct{\cS}=\frac{1}{8} (r_{i} -q_{i'}) (r_{k} -q_{k'} +2 (q_{j} -q_{j'})).
\end{equation*}
Note that, by \eqref{eq:Mm} and~\eqref{eq:ABC}
\[r_{i} -q_{i'}\geq 2(M+m^2) - 2(M+m+b)=2(m^2-m-b)>0\]
and
\[r_{k} -q_{k'} +2 (q_{j} -q_{j'})\geq 2(M+m^2)-2(M+m+b) + 2(2M-2(M+m+b))=2(m^2-3m-3b)>0.\]
Thus, $\tct{\cS'}>\tct{\cS}$ which contradicts the optimality of $\cS$.
This proves that each shared processor executes exactly one job from the set $C$.

\medskip
Third, we prove that each shared processor does exactly one job from $A$.
Analogously as before, consider an optimal schedule $\cS$ in which some shared processor $\Mshared_l$ executes a job $k\in C$ 
and jobs 
$i,j\in A$ and some other shared processor $\Mshared_{l'}$ executes a job 
$k'\in C$ and no job from $A$ (thus, the two remaining jobs on that processor
$i',j'\in B$).
By~\eqref{eq:ABC}, the job $k$ is longer that the jobs $i$ and $j$, and similarly, the job $k'$ is longer than $i'$ and $j'$.
By Lemma~\ref{lem:Vshape}, the schedule $\cS$ is V-shaped and thus $k$ is the first or the last job on $\Mshared_l$ and $k'$ is the first or the last job on $\Mshared_{l'}$.
Furthermore, Lemma~\ref{lem:reversability} implies that we may without loss of generality assume that $k$ and $k'$ are the last jobs on $\Mshared_l$ and $\Mshared_{l'}$, respectively.
Let $i\in A$ be the first job on $\Mshared_l$, and $i'\in B$ be the first job on $\Mshared_{l'}$.
We have $\tct{\cS}=\sigma-\xi(\mathbf{A})/2-\xi(\mathbf{A}')/2$, where
\[\mathbf{A}=\begin{bmatrix}
0 & \frac{1}{4} s_{i} s_{j} & \frac{1}{8} s_{i} r_{k} \\
\frac{1}{4}  s_{j} s_{i} & 0 & \frac{1}{4} s_{j} r_{k} \\
\frac{1}{8}  r_{k} s_{i} & \frac{1}{4} r_{k} s_{j} & 0
\end{bmatrix},
\quad
\mathbf{A}'=\begin{bmatrix}
0 & \frac{1}{4} b_{i'} b_{j'} & \frac{1}{8} b_{i'} r_{k'} \\
\frac{1}{4} b_{j'}  b_{i'} & 0 & \frac{1}{4} b_{j'} r_{k'} \\
\frac{1}{8} r_{k'}  b_{i'}  & \frac{1}{4} r_{k'} b_{j'} & 0
\end{bmatrix}
\]
and $\sigma=\sum_{i\neq l, l'} \tct{\cS_i}+\frac{1}{2}\vectP_l\cdot\vectI_3\cdot\vectP_l^{\T}+\frac{1}{2}\vectP_{l'}\cdot\vectI_3\cdot\vectP_{l'}^{\T}$.
Obtain a schedule $\cS'$ by exchanging in $\cS$ the $i\in A$ from $\Mshared_l$ with the $i'\in B$ from $\Mshared_{l'}$.
Observe that the first jobs on $\Mshared_{l}$ and $\Mshared_{l'}$ complete by $(M+m^2+b)$ in $\cS'$, the second jobs on those processors  complete by $\frac{3}{2}(M+m^2+b)$, and moreover the shortest job in the instance is not shorter than $2M$, thus $\frac{3}{2}(M+m^2+b)<2M$ for $M>7(m^2+b)$, according to~\eqref{eq:Mm}, and consequently all jobs on $\Mshared_{l}$ and $\Mshared_{l'}$ are long enough to be executed on $\Mshared_{l}$ and $\Mshared_{l'}$ after the exchange.
This implies that $\cS'$ is feasible.
For the new schedule we have $\tct{\cS'}=\sigma-\xi(\mathbf{B})/2-\xi(\mathbf{B}')/2$, where
\[\mathbf{B}=\begin{bmatrix}
0 & \frac{1}{4} b_{i'} s_{j} & \frac{1}{8} b_{i'} r_{k} \\
\frac{1}{4} b_{i'} s_{j} & 0 & \frac{1}{4} s_{j} r_{k} \\
\frac{1}{8} b_{i'} r_{k} & \frac{1}{4} r_{k} s_{j} & 0
\end{bmatrix},
\quad
\mathbf{B}'=\begin{bmatrix}
0 & \frac{1}{4} s_{i} b_{j'} & \frac{1}{8} s_{i} r_{k'} \\
\frac{1}{4} s_{i} b_{j'} & 0 & \frac{1}{4} b_{j'} r_{k'} \\
\frac{1}{8} s_{i} r_{k'} & \frac{1}{4} r_{k'} b_{j'} & 0
\end{bmatrix}.
\]
Therefore,\begin{equation*}
\tct{\cS'}-\tct{\cS}=\frac{1}{8} (s_{i} -b_{i'}) (r_{k} -r_{k'} +2 (s_{j} -b_{j'})).
\end{equation*}
We have $s_i-b_{i'}\geq 2m-b>0$ and $r_{k} -r_{k'} +2 (s_{j} -b_{j'})\geq 4m-4b>0$, where both inequalities follow from~\eqref{eq:Mm} and~\eqref{eq:ABC}.
Therefore, we obtain $\tct{\cS'}>\tct{\cS}$ --- a contradiction.
This proves that each shared processor does exactly one job from each set $A ,B$ and $C$.

\medskip
It remains to argue that for three jobs $i\in A$, $j\in B$ and $k\in C$ scheduled on some shared processor $\Mshared_l$, their order on $\Mshared_l$ is $(i,j,k)$.
By Lemmas~\ref{lem:Vshape} and~\ref{lem:reversability} and~\eqref{eq:ABC}, possible orders in an optimal synchronized schedule are $(i,j,k)$, $(j,i,k)$, and take two schedules $\cS$ and $\cS'$ that execute the jobs in these orders, respectively.
We have by \eqref{eq:Mm}
\[\tct{\cS}-\tct{\cS'} = \frac{r_k}{8}(s_i-b_j)>0,\]
which completes the proof of the lemma.
\end{proof}

We conclude this section with its main result.
\begin{theorem} \label{thm:NPC}
The weighted multiple shared-processors problem \textup{WSMP} is strongly \textup{NP}-hard.
\end{theorem}
\begin{proof}
Note that the weights and the processing times of jobs in our reduction are bounded by $O(M+m^2+b)$.
By \eqref{eq:Mm}, $M$ and $m$ are polynomially bounded by $b$ and the value of $b$ is bounded by a polynomial in $n$ since the \textup{N3DM} problem is strongly NP-hard.
Thus, the theorem follows from Lemmas~\ref{lem:NPCiff} and~\ref{lem:equitable}.
\end{proof}

\section{An $O(n\log n)$ algorithm for equal weights}

This section gives an $O(n\log n)$ optimization algorithm for the \textup{WSMP} problem with equal weights, i.e. $w_{i}=w$ for $i\in \jobs$.
Without loss of generality we assume $w=1$ for convenience. We begin with the following result of \cite{VairaktarakisAydinliyim07} for a single shared processor non-preemptive problem and extended to preemptive one in \cite{HK15}.

\begin{lemma}[\cite{HK15,VairaktarakisAydinliyim07}]
\label{lem:HK15}
If jobs $1,\ldots, n$ with unit weights and processing times $p_1,\ldots,p_n$, respectively, are executed on a shared processor in an optimal synchronized schedule in the order $1,\ldots,n$,
then $p_1\leq p_2\leq \cdots \leq p_n$ and the total weighted overlap equals
\[ \sum_{i=1}^n \bar{t_i}=\frac{p_n}{2} + \frac{p_{n-1}}{4} + \cdots + \frac{p_1}{2^n}.\]
\qed
\end{lemma}

This hints at the following 
algorithm for the \textup{WSMP} problem with unit weights. Take the following sequence of positional weights:%
\begin{equation} \label{position}
\underbrace{\frac{1}{2},\ldots,\frac{1}{2}}_{m-times},\underbrace{\frac{1}{4}%
,\ldots,\frac{1}{4}}_{m-times},\ldots,\underbrace{\frac{1}{2^{k}},\ldots,\frac{1}{%
2^{k}}}_{r-times}
\end{equation}
where 
\begin{equation}\label{a}
k=\left\lceil \frac{n}{m}\right\rceil \text{ and } r=n-\left\lfloor \frac{n}{m} \right\rfloor m,
\end{equation}
and order the jobs in descending order of their processing times so that%
\begin{equation} \label{proctime}
p_{n}\geq \cdots\geq p_{1}. 
\end{equation}
Match the $i$-th positional weight from the left in the sequence \eqref{position} with the the $i$-th job form the left in the sequence \eqref{proctime} for $i\in\{1,\ldots,n\}$. Partition the set of jobs
$\jobs$ into $m$ disjoint subsets 
\[
\jobs_{1},\ldots,\jobs_{m} 
\]%
so that any two jobs in a subset are matched with different positional weights. Thus, in each subset we have exactly one job matched with $\frac{1}{2}$, exactly one with $\frac{1}{4}$, \ldots, and exactly one matched with $\frac{1}{2^{k-1}}$. Moreover, there are exactly $0\leq r<m$ subsets with exactly one job matched with $\frac{1}{2^{k}}$ in each. Without loss of generality we may assume that these $r$ sets are $\jobs_{1},\ldots,\jobs_{r}$. Finally, schedule the jobs from $\jobs_{\ell}$ on shared processor $\mathcal{M}_{\ell}$ in ascending order of their processing times for each $\ell\in\{1,\ldots,m\}$. Let the resulting synchronized schedule be $\bar{\cS}$, and let $\bar{\mathcal{S}}_{\ell}$ be the synchronized schedule for $\Mshared_{\ell}$ and the private processors of jobs executed on $\Mshared_{\ell}$ for each $\ell\in\{1,\ldots,m\}$. From  Lemma~\ref{lem:HK15} and this algorithm we immediately obtain.

\begin{lemma}\label{S}
It holds that
\[\tct{\bar{\cS}}=\sum_{\ell=1}^m \tct{\bar{\cS}_{\ell}}=\sum_{\ell=1}^r \sum _{i=1}^{\lceil \frac{n}{m} \rceil} \frac{p^{\ell}_{i}}{2^{\lceil \frac{n}{m} \rceil+1-i}}+ \sum_{\ell=r+1}^{m} \sum _{i=1}^{\lfloor \frac{n}{m} \rfloor} \frac{p^{\ell}_{i}}{2^{\lfloor \frac{n}{m} \rfloor+1-i}},\]
where $\jobs_{\ell}=\{p^{\ell}_{i}\st i\in\{1,\ldots,|\jobs_{\ell}|\}\}$ and $p^{\ell}_1\leq\cdots\leq p^{\ell}_{|\jobs_{\ell}|}$ for each $\ell\in\{1,\ldots,m\}$.
\qed
\end{lemma}

It remains to show that  $\bar{\cS}$ is optimal.
We start by noting that in optimal schedules for the \textup{WSMP} with unit weights each job executes on some shared processor.
\begin{lemma} \label{lem:unit-all}
Let $\cS$ be an optimal schedule for a set of jobs $\jobs$ with unit weights and $m$ shared processors.
Then, each job executes on some shared processor.
\end{lemma}
\begin{proof}
Suppose that some schedule $\cS$ does not satisfy the lemma.
Take an arbitrary shared processor $\Mshared_{\ell}$ and any job $j\in\jobs$ with processing time $p$ that executes entirely on its private processor $\Mpriv_j$.
Suppose that jobs $j_1,\ldots,j_k$ with processing times $p_1,\ldots,p_k$, respectively, execute on $\Mshared_{\ell}$ in $\cS$ in this order.
By Lemma~\ref{lem:HK15}, $p_1\leq\cdots\leq p_k$.
Take the maximum $i\in\{1,\ldots,k\}$ such that $p_i\leq p$. If $p<p_1$, then take $i=0$.
Consider a synchronized schedule $\cS'$ that is identical to $\cS$ on all shared processors different than $\Mshared_{\ell}$ and executes jobs $j_1,\ldots,j_i,j,j_{i+1},\ldots,j_k$ (with processing times $p_1,\ldots,p_i$, $p,p_{i+1},\ldots,p_k$ respectively), in this order, on $\Mshared_{\ell}$.
Due to the choice of $i$, $\cS'$ is feasible and by Lemma~\ref{S} for $i\geq1$
\begin{eqnarray*}
\tct{\cS'}-\tct{\cS} & = & \sum_{\ell=1}^{i}\frac{p_{\ell}}{2^{k+2-\ell	}}+\frac{p}{2^{k+1-i}} - \sum_{\ell=1}^{i}\frac{p_{\ell}}{2^{k+1-\ell}} \\
                     & = & \frac{p_1}{2^{k+1}} + \sum_{\ell=2}^{i}\frac{p_{\ell}-p_{\ell-1}}{2^{k+2-\ell}} + \left(\frac{p-p_i}{2^{k+1-i}}\right).
\end{eqnarray*}
Since $p_1\leq\cdots\leq p_i\leq p$, we obtain that $\tct{\cS'}-\tct{\cS}\geq p_1/2^{k+1}>0$, which implies that $\cS$ is not optimal and completes the proof for $i\geq 1$. Finally, 
\begin{eqnarray*}
\tct{\cS'}-\tct{\cS} & = & \frac{p}{2^{k+1}},  
\end{eqnarray*}
for $i=0$ which implies that $\cS$ is not optimal and completes the proof.
\end{proof}
\begin{lemma} \label{opt}
$\bar{\mathcal{S}}$ is optimal.
\end{lemma}
\begin{proof}
Let $\mathcal{S'}$ be a synchronized schedule with jobs $\jobs'_1, \ldots, \jobs'_m$ on shared processors $\mathcal{M}_{1}, \ldots, \mathcal{M}_{m}$, respectively.
Denote $n_{\ell}=|\jobs'_{\ell}|$ and denote the processing times of jobs in $\jobs'_{\ell}$ by $q_1^{\ell} \leq\cdots\leq q_{n_{\ell}}^{\ell}$, $\ell\in\{1,\ldots,m\}$.
We have by Lemma~\ref{lem:HK15},
\[\tct{\cS'}\leq \sum_{\ell=1}^{m}\sum_{i=1}^{n_{\ell}}\frac{q_{i}^{\ell}}{2^{n_{\ell}+1-i}}.\]
We will argue that 
\begin{equation} \label{eq}
\sum_{\ell=1}^{m}\sum_{i=1}^{n_{\ell}}\frac{q_{i}^{\ell}}{2^{n_{\ell}+1-i}} \leq \sum_{\ell=1}^r \sum _{i=1}^{\lceil \frac{n}{m} \rceil} \frac{p^{\ell}_{i}}{2^{\lceil \frac{n}{m} \rceil+1-i}}+ \sum_{\ell=r+1}^{m} \sum _{i=1}^{\lfloor \frac{n}{m} \rfloor} \frac{p^{\ell}_{i}}{2^{\lfloor \frac{n}{m} \rfloor+1-i}}
\end{equation}
which by Lemma \ref{S} proves the lemma. To prove inequality \eqref{eq} take the positional weights of the left hand side of \eqref{eq} in the non-ascending order. Let them make a vector $\vecta$. By Lemma~\ref{lem:unit-all}, the length of $\vecta$ is $n$.
We obtain a vector $\vecta'$ as follows.
Initially set $\vecta':=\vecta$ and perform the following action as long as possible.
Find the minimum $i\in\{1,\ldots,k-1\}$ (recall (\ref{a}) for definition of $k$ and $r$) such that the value $1/2^i$ appears less than $m$ times in $\vecta'$, and replace any entry of $\vecta'$ with value less than $1/2^i$ with the value $1/2^i$.
Finally, any value less than $1/2^k$ replace in $\vecta'$ with $1/2^k$.
Clearly, $1/2^i$ appears exactly $m$ times in $\vecta'$ for each $i\in\{1,\ldots,k-1\}$, and $1/2^k$ appears exactly $r$ times in $\vecta'$.

Take the positional weights of the right hand side of (\ref{eq}) in the non-ascending order. Let them make a vector $\vectb$.
We observe that $\vecta'$ is a permutation of $\vectb$ and thus we
can readily show that
\[\sum_{i=1}^{\ell} \vecta_i\leq\sum_{i=1}^{\ell} \vecta_i'\leq \sum_{i=1}^{\ell} \vectb_i\]
for each $\ell\in\{1,\ldots,n\}$. Hence the inequality (\ref{eq}) holds by the rearrangement inequality of Hardy-Littlewood-Polya \cite{HLP}.
\end{proof}

We conclude this section with the following result.
\begin{theorem}
For any set of jobs $\jobs$ with equal weights and arbitrary processing times and $m\geq 1$ shared processors, the schedule $\bar{\cS}$ is an optimal solution to the \textup{WSMP} problem and can be computed in $O(n\log n)$-time.
\qed
\end{theorem}

\section{Conclusions and open problems}

We studied the shared multi-processor scheduling problem. We proved that the maximization of total weighted overlap is NP-hard in the strong sense 
though the case with equal weights is solvable in $O(n\log n)$ time. We also proved that synchronized schedules include optimal schedules. This characterization
as well as  other characteristics of special subclasses of the problem may prove instrumental in settling the complexity of  the single processor case, which remains
open, and in developing efficient branch-and-bound algorithms,  heuristics, and approximation algorithms with guaranteed worst case approximation for the problem.
We conjecture that the single processor case is NP-hard even for instances with processing times equal weights for all jobs.

The design of coordinating mechanisms to ensure efficiency of decentralized shared multi-processor scheduling remains an interesting line of research in supply chains
with subcontracting.  In particular, coordinating pricing schemes for multi-processor problem with equal weights (such schemes do not exist for the weighted case in general, see \cite{HK15}) seem to readily extend those developed in  \cite{HK15} for a single processor.

In this paper we assumed that a job can use only a single shared processor, if any. However, relaxations of this assumption that allow for using an arbitrary or a fixed number of shared processors by a job 
could possibly lead to interesting scheduling problems. We leave investigation of these relaxations for further research. 

\section*{Acknowledgements}
This research has been supported by the Natural Sciences and Engineering Research
Council of Canada (NSERC) Grant OPG0105675.
Dariusz Dereniowski was partially supported by Polish National Science Center under contract DEC-2011/02/A/ST6/00201.

\bibliographystyle{plain}
\bibliography{references}
\end{document}